%% file: EPI_v3.tex
\documentclass[11pt]{article}

\usepackage[margin=1in]{geometry}

\usepackage{mathpazo}

\usepackage{microtype}

\usepackage{tikz}

\usepackage[colorlinks = true]{hyperref}
\hypersetup{
  pdftitle = {Entropy power inequalities for qudits},
  pdfauthor = {Koenraad Audenaert, Nilanjana Datta, Maris Ozols}
}
\definecolor{darkred}  {rgb}{0.5,0,0}
\definecolor{darkblue} {rgb}{0,0,0.5}
\definecolor{darkgreen}{rgb}{0,0.5,0}
\hypersetup{
  urlcolor   = blue,         
  linkcolor  = darkblue,     
  citecolor  = darkgreen,    
  filecolor  = darkred       
}

\usepackage{amsmath,amssymb,amsfonts,amsthm,amstext}
\numberwithin{equation}{section}

\usepackage{thm-restate}

\usepackage{mathtools}
\mathtoolsset{centercolon}
\makeatletter
\protected\def\tikz@nonactivecolon{\ifmmode\mathrel{\mathop\ordinarycolon}\else:\fi}
\makeatother

\usepackage{cleveref}
\crefname{lemma}{Lemma}{Lemmas}
\crefname{definition}{Definition}{Definitions}
\crefname{theorem}{Theorem}{Theorems}
\crefname{conjecture}{Conjecture}{Conjectures}
\crefname{corollary}{Corollary}{Corollaries}
\crefname{section}{Section}{Sections}
\crefname{appendix}{Appendix}{Appendices}
\crefname{figure}{Fig.}{Figs.}
\crefname{table}{Table}{Tables}

\usepackage[retainorgcmds]{IEEEtrantools}

\usepackage{multirow}

\newcommand{\ket}[1]{|#1\rangle}
\newcommand{\bra}[1]{\langle#1|}
\newcommand{\braket}[2]{\langle#1|#2\rangle}
\newcommand{\proj}[1]{|#1\rangle\langle#1|}

\newcommand{\x}{\otimes}

\newcommand{\ct}{^{\dagger}}

\DeclarePairedDelimiter{\set}{\lbrace}{\rbrace}
\DeclarePairedDelimiter{\abs}{\lvert}{\rvert}

\DeclarePairedDelimiter{\of}{\lparen}{\rparen}
\DeclarePairedDelimiter{\sof}{\lbrack}{\rbrack}

\DeclareMathOperator{\diag}{diag}

\DeclareMathOperator{\Tr}{Tr}

\newcommand{\mx}[1]{\begin{pmatrix}#1\end{pmatrix}}

\newcommand{\C}{\mathbb{C}}
\newcommand{\R}{\mathbb{R}}
\newcommand{\N}{\mathbb{N}}

\newcommand{\D}[1]{\mathcal{D}(\C^{#1})}             
\newcommand{\DD}[1]{\mathcal{D}(\C^{#1} \x \C^{#1})} 

\newcommand{\mc}[1]{\mathcal{#1}}
\newcommand{\F}{\mc{F}} 
\newcommand{\E}{\mc{E}} 
\newcommand{\tE}{\widetilde{\E}} 

\newcommand{\cH}{\mathcal{H}}

\newcommand{\cD}{\mathcal{D}}
\newcommand{\cU}{\mathcal{U}}
\newcommand{\cL}{\mathcal{L}}
\newcommand{\cK}{\mathcal{K}}
\newcommand{\cT}{\mathcal{T}}

\newcommand{\dar}{\downarrow}

\newcommand{\trho}{\tilde{\rho}}
\newcommand{\tsigma}{\tilde{\sigma}}

\renewcommand{\a}{a} 
\newcommand{\cs}{c}  
\newcommand{\EP}{E}  
\renewcommand{\L}{L} 
\renewcommand{\v}{u} 

\newcommand{\f}{\ell}

\newtheorem{theorem}{Theorem}
\newtheorem{lemma}[theorem]{Lemma}
\newtheorem{definition}[theorem]{Definition}
\newtheorem{corollary}[theorem]{Corollary}

\newtheorem*{example}{Example}
\newtheorem*{problem}{Problem}

\theoremstyle{definition}
\newtheorem*{remark}{Remark}

\begin{document}

\title{Entropy power inequalities for qudits}

\author{Koenraad Audenaert\thanks{Department of Mathematics, Royal Holloway University of London, Egham TW20 0EX, UK
\& Department of Physics and Astronomy, Ghent University, S9, Krijgslaan 281, B-9000 Ghent, Belgium}
\and Nilanjana Datta\thanks{Statistical Laboratory, Centre for Mathematical
Sciences, University of Cambridge, Wilberforce Road, Cambridge CB3 0WB, UK}
\and Maris Ozols\thanks{Department of Applied Mathematics and Theoretical Physics, University of Cambridge, Cambridge CB3 0WA, UK}}

\maketitle

\begin{abstract}
Shannon's entropy power inequality (EPI) can be viewed as a statement of concavity of an entropic function of a continuous random variable under a scaled addition rule:
$$f(\sqrt{\a}\,X + \sqrt{1-\a}\,Y) \ge \a f(X) + (1-\a) f(Y) \quad \forall \, \a \in [0,1].$$
Here, $X$ and $Y$ are continuous random variables and the function $f$ is either the differential entropy or the \emph{entropy power}. K\"onig and Smith [\href{https://dx.doi.org/10.1109/TIT.2014.2298436}{\textit{IEEE Trans.\ Inf.\ Theory.}\ 60(3):1536--1548, 2014}] and De Palma, Mari, and Giovannetti [\href{https://dx.doi.org/10.1038/nphoton.2014.252}{\textit{Nature Photon.} 8(12):958--964, 2014}] obtained quantum analogues of these inequalities for continuous-variable quantum systems, where $X$ and $Y$ are replaced by bosonic fields and the addition rule is the action of a beamsplitter with transmissivity $a$ on those fields. In this paper, we similarly establish a class of EPI analogues for $d$-level quantum systems (i.e.~qudits). The underlying addition rule for which these inequalities hold is given by a quantum channel that depends on the parameter $\a \in [0,1]$ and acts like a finite-dimensional analogue of a beamsplitter with transmissivity $a$, converting a two-qudit product state into a single qudit state. We refer  to this channel as a \emph{partial swap channel} because of the particular way its output interpolates between the states of the two qudits in the input as $\a$ is changed from zero to one. We obtain analogues of Shannon's EPI, not only for the von Neumann entropy and the entropy power for the output of such channels, but for a much larger class of functions as well. This class includes the R\'enyi entropies and the subentropy. We also prove a qudit analogue of the entropy photon number inequality (EPnI). Finally, for the subclass of partial swap channels for which one of the qudit states in the input is fixed, our EPIs and EPnI yield lower bounds on the minimum output entropy and upper bounds on the Holevo capacity.
\end{abstract}


\section{Introduction}

Inequalities between entropic quantities play a fundamental role in information theory and have been employed effectively in finding bounds on optimal rates of various information-processing tasks. Shannon's entropy power inequality (EPI)~\cite{Shannon} is one such inequality and it has proved to be of relevance in studying problems not only in information theory, but also in probability theory and mathematical physics~\cite{Stam}. It has been used, for example, in finding upper bounds on the capacities of certain noisy channels (e.g.~the Gaussian broadcast channel~\cite{Bergmans}) and in proving convergence in relative entropy for the Central Limit Theorem~\cite{CLT}.

\subsubsection*{Classical EPIs}

For an arbitrary random variable $X$ on $\R^d$ with probability density function (p.d.f.) $f_X$, the \emph{entropy power} of $X$ is the quantity
\begin{equation}
  v(X) := \frac{e^{2H(X)/d}}{2\pi e},
  \label{eq:pdf ep}
\end{equation}
where $H(X)$ is the \emph{differential entropy} of $X$,
\begin{equation}
  H(X):= - \int_{\R^d} f_X(x) \log f_X(x) dx
\end{equation}
(throughout the paper we use $\log$ to represent the natural logarithm). The name ``entropy power'' is derived from the following fact: if $X$ is a Gaussian random variable on $\R$ with zero mean and variance $\sigma^2$, then $H(X) =(1/2) \log (2 \pi e \sigma^2)$; hence $v(X)$ is equal to its variance, which is commonly referred to as its \emph{power}. Note that the entropy power of a random variable $X$ is equal to the variance of a Gaussian random variable which has the same differential entropy as $X$. For $X$ on $\R^d$, we shall henceforth omit the factor $1/2\pi e$ and refer to $e^{2H(X)/d}$ as the entropy power, as in~\cite{KS14}.

The entropy power satisfies the following \emph{scaling property}: $v(\sqrt{\alpha} X) = \alpha v(X)$. This follows from the scaling property of p.d.f.s: if $f_{\alpha X}$ denotes the p.d.f.\ of a random variable $\alpha X$ on $\R^d$, where $\alpha > 0$, then $f_{\alpha X}(x) = \alpha^{-d} f_X(x/\alpha)$, $x \in \R^d$, which in turn implies that $H(\alpha X) = H(X) + d \log \alpha$. This shows why the factor $1/d$ in the definition of $v(X)$ has to be there for $X$ on $\R^d$.

Shannon's EPI~\cite{Shannon} provides a lower bound on the entropy power of a sum of two independent random variables $X$ and $Y$ on $\R^d$ in terms of the sums of the entropy powers of the individual random variables:
\begin{equation}\label{vEPI}
  v(X+Y) \geq v(X) + v(Y),
\end{equation}
or equivalently,
\begin{equation}\label{cEPI}
  e^{2H\of{X+ Y}/d} \geq e^{2H(X)/d} + e^{2H(Y)/d}.
\end{equation}
Here, $H(X+Y)$ is the differential entropy of the p.d.f.\ of the sum $Z := X + Y$, which is given by the convolution
\begin{equation}\label{add1}
  f_{X+Y}(x) = (f_X\ast f_Y)(x):= \int_{\R^d} f_X(x') f_Y(x-x') dx', \qquad \forall \, x \in \R^d.
\end{equation}

The inequality \cref{vEPI} was proposed by Shannon in~\cite{Shannon} as a means to bound the capacity of a non-Gaussian \emph{additive noise channel}, that is, a channel with input $X$ and output $X+Y$, with $Y$ being an independent (non-Gaussian) random variable modeling the noise which is added to the input. Later, Lieb~\cite{Lieb} and Dembo, Cover, and Thomas~\cite{DCT91} (see also~\cite{VG06}) showed that the EPI~\eqref{cEPI} can be equivalently expressed as the following inequality between differential entropies:
\begin{equation}\label{eq:cEPI-a}
  H\of[\big]{ \sqrt{\a} \, X + \sqrt{1-\a} \, Y }
  \geq \a H(X) + (1-\a) H(Y),
  \qquad \forall \, \a \in [0,1].
\end{equation}
The above inequality was proved by employing the R\'enyi entropy~\cite{Renyi} and using properties of $p$-norms on convolutions given by a sharp form of Young's inequality~\cite{Beckner}.

The form of the EPI in \cref{eq:cEPI-a} motivates the definition of an operation (which following~\cite{KS13,KS14} we denote as $\boxplus_\a$) on the space of random variables, given by the following \emph{scaled addition rule}:
\begin{equation}\label{eq:c-add}
  X \boxplus_\a Y := \sqrt{\a} \, X + \sqrt{1-\a} \, Y, \qquad \forall \, \a \in [0,1].
\end{equation}
The random variable $X \boxplus_\a Y$ can be interpreted as an interpolation between $X$ and $Y$ as $\a$ is decreased from $1$ to $0$. With this notation, the inequality \eqref{eq:cEPI-a} can be written as
\begin{equation}\label{cEPI1}
  H\of{X \boxplus_\a Y} \geq \a H(X) + (1-\a) H(Y), \qquad \forall \, \a \in [0,1].
\end{equation}
Using the scaling property of the entropy power, the EPI~\eqref{cEPI} can be expressed as follows:
\begin{equation}\label{cEPI2}
  e^{2H\of{X \boxplus_{\a} Y}/d} \geq \a e^{2H(X)/d} + (1-\a) e^{2H(Y)/d}, \qquad \forall \, \a \in [0,1].
\end{equation}
Shannon's EPI \eqref{cEPI} (and hence also \eqref{cEPI1} and \eqref{cEPI2}) was first proved rigorously by Stam~\cite{Stam} and by Blachman~\cite{Blachman}, by employing de Bruijn's identity, which couples Fisher information with differential entropy. Since then various different proofs and generalizations of the EPI have been proposed (see e.g.~\cite{VG06, Rioul, SS15} and references therein).

It is natural to conjecture that an analogue of Shannon's EPI also holds for discrete random variables e.g.~on non-negative integers. This conjecture was first proved by~\cite{HV03} for the case of binomial random variables. They proved that if $X_n \sim \mathrm{Bin}(n, p)$, then for $p=1/2$ (see also~\cite{SDM11}):
\begin{equation}\label{cEPI-bin}
  e^{2H\of{X_n+ X_m}} \geq e^{2H(X_n)} + e^{2H(X_m)}, \qquad \forall \, m,n \geq 1.
\end{equation}
Further, Johnson and Yu~\cite{JY10} established a form of the EPI which is valid for ultra log-concave discrete random variables (see Definition 2.2 of~\cite{JY10}), whereby the scaling operation of a continuous random variable was suitably replaced by the so-called {\em{thinning operation}} introduced by R\'enyi~\cite{thin}, which is considered to be an analogue of scaling for discrete random variables.

\subsubsection*{Quantum analogues of EPIs}

The discovery of an analogue of the EPI in the quantum setting by K\"onig and Smith~\cite{KS14} marked a significant advance in quantum information theory. They proposed an EPI which holds for continuous-variable quantum systems that arise, for example, in quantum optics. In this case, the random variables $X$, $Y$ of the classical EPIs (\cref{cEPI1} and \cref{cEPI2}) are replaced by quantum fields, bosonic modes of electromagnetic radiation, described by quantum states $\rho_X$, $\rho_Y$, which act on a separable, infinite-dimensional Hilbert space $\cH$. The differential entropy is accordingly replaced by the \emph{von~Neumann entropy} $H(\rho):= - \Tr(\rho \log \rho)$.

A prerequisite for any quantum analogue of the EPI is the formulation of a suitable analogue of the addition rule \eqref{eq:c-add} which can be applied to pairs of quantum states. Since the quantum-mechanical analogue of additive noise can be modelled by the mixing of two beams of light at a beamsplitter, K\"onig and Smith considered the parameter $\a$ in \cref{eq:c-add} to be the beamsplitter's transmissivity. The classical addition rule \cref{eq:c-add} is thereby replaced by an analogous \emph{quantum field addition rule} for the field operators. In particular, if the two input signals are $m$-mode bosonic fields, with annihilation operators $\hat{a}_1, \dotsc, \hat{a}_m$ and $\hat{b}_1, \dotsc, \hat{b}_m$ respectively, then the output is an $m$-mode bosonic field with annihilation operators $\hat{c}_1, \dotsc, \hat{c}_m$, where
\begin{equation}\label{eq:cs}
  \hat{c}_i
  := \sqrt{a} \, \hat{a}_i + \sqrt{1-a} \, \hat{b}_i.
\end{equation}
In a state space description, the input signals are described by quantum states $\rho_X$, $\rho_Y$ on $\cH$. This yields an equivalent \emph{quantum state addition rule}, where the beamsplitter converts the incoming state $\rho_X \otimes \rho_Y$ to a state $\rho_X \boxplus_\a \rho_Y$ given by
\begin{equation}\label{eq:q-add}
  (\rho_X, \rho_Y)
  \mapsto \rho_X \boxplus_\a \rho_Y
  := \E_\a (\rho_X \x \rho_Y).
\end{equation}
Here, $\E_\a$ is a linear, completely positive trace-preserving map defined through the relation 
\begin{equation}
  \E_\a(\rho_{XY}) := \Tr_Y (U_\a \rho_{XY} U_\a^\dagger), 
\end{equation}
with the partial trace being taken over the second system, and $U_a$ is the unitary operator describing the action of the beamsplitter on the state space $\cH$. Analogous to the classical case, the state $\rho_X \boxplus_\a \rho_Y$ reduces to $\rho_X$ when $\a=1$, and to $\rho_Y$ when $\a=0$.

K\"onig and Smith~\cite{KS14} proved that the following quantum analogues of the EPIs \eqref{cEPI1} and \eqref{cEPI2} hold, under the quantum addition rule given by \cref{eq:q-add}:
\begin{align}
  H(\rho_X \boxplus_\a \rho_Y)
    &\geq \a H(\rho_X) + (1-\a) H(\rho_Y), \label{qEPI1} \\
  e^{H(\rho_X \boxplus_{1/2} \rho_Y)/m}
    &\geq \frac{1}{2} e^{H(\rho_X)/m} + \frac{1}{2} e^{H(\rho_Y)/m}, \label{qEPI1/2}
\end{align}
where $m$ is the number of bosonic modes. The inequality \eqref{qEPI1/2} corresponds to a 50\,:\,50 beamsplitter (i.e., a beamsplitter with transmissivity $\a=1/2$). Later, De~Palma \textit{et al.}~\cite{DMG14} proved that an analogous inequality also holds for \emph{any} beamsplitter (i.e., for any $\a \in [0,1]$) and is given by the following:
\begin{equation}
  e^{H(\rho_X \boxplus_\a \rho_Y)/m}
  \geq \a e^{H(\rho_X)/m} + (1-\a) e^{H(\rho_Y)/m},
  \qquad \forall \, \a \in [0,1].
  \label{qEPI2}
\end{equation}
Note that the EPI given by \cref{qEPI1/2} seems to differ from its classical counterpart \eqref{cEPI2} by a factor of $2$ in the exponent. However, one can argue that the dimension of the bosonic phase space is $d=2m$ (as there are 2 quadratures per mode). These EPIs have found applications for bounding classical capacity of bosonic channels~\cite{KS13,KS13bits}.

The above inequalities do not reduce to the classical EPIs \eqref{cEPI1} and \eqref{cEPI2} for commuting states; in other words, they are not quantum generalizations of the Shannon's original EPI in the usual sense, as they do not include the latter as a special case. This is because the addition rule acts at the field operator level and not at the state level. In fact, the dependence of the output state on the parameter $a$ is much more complicated than in the classical case.

Another inequality, related to the EPI \eqref{qEPI2}, was conjectured by Guha \emph{et al}.~\cite{GES08} and is known as the {\em{entropy photon number inequality}} (EPnI). The thermal state of a bosonic mode with annihilation operator $\hat{a}$ can be expressed as~\cite{GSE08}:
\begin{equation}
  \rho_T = \sum_{i=0}^\infty \frac{N^i}{(N+1)^{i+1}} \ket{i}\bra{i},
\end{equation}
where $N:= \Tr \of{ \rho_T \, \hat{a}^\dagger \hat{a} }$ is the average photon number of the state $\rho_T$. Its von Neumann entropy can be evaluated as $H(\rho_T) = g(N)$ where $g(x) := (1+x) \log (1+x) - x \log x$. Inverting this, the photon number of $\rho_T$ is then $N = g^{-1} \of{H(\rho_T)}$. Correspondingly, the {\em{photon number}} of an $m$-mode bosonic state $\rho$ is defined as $N(\rho) := g^{-1} \of{H(\rho)/m}$. Guha \emph{et al}.~\cite{GES08} conjectured that
\begin{equation}\label{EPnI}
  N(\rho_X \boxplus_\a \rho_Y) \geq a N(\rho_X) + (1-a) N(\rho_Y), \qquad \forall \, a \in [0,1],
\end{equation}
where $\boxplus_a$ is again the quantum state addition rule \eqref{eq:q-add}. This conjecture is of particular significance in quantum information theory since if it were true then it would allow one to evaluate classical capacities of various bosonic channels, e.g.\ the bosonic broadcast channel~\cite{GSE07} and the wiretap channel~\cite{GSE08}. It has thus far been proved only for Gaussian states~\cite{Guha}.

\medskip

A natural question to ask is whether quantum EPIs can also be found outside the continuous-variable setting. In this paper, we address this question by formulating an addition rule for $d$-level systems (qudits) in the form of a quantum channel $\E_a$, which we call the \emph{partial swap channel}, that acts on the two input quantum states. We then prove analogues of the quantum EPIs \eqref{qEPI1} and \eqref{qEPI2} for this addition rule. We also prove similar inequalities for a large class $\F$ of functions, including the R\'enyi entropies of order $\alpha \in [0,1)$ and the subentropy~\cite{JRW94}. Again these are analogues and not generalizations of the classical EPIs for discrete random variables~\cite{HV03, JY10, SDM11} to the non-commutative setting, as the latter do not emerge as special cases for commuting states.

Furthermore, the concept of entropy photon number $N$ has a straightforward generalization to qudit systems via its one-to-one relation with the von Neumann entropy, $H=g(N)$, even though it loses its interpretation as an average photon number. We show that the function $g^{-1}$ is in the class $\F$, and as a result obtain the EPnI for our qudit addition rule.

Finally, we apply our results (EPIs and EPnI) to obtain lower bounds on the minimum output entropy and upper bounds on the Holevo capacity for a class of single-input channels that are formed from the channel $\E_a$ by fixing the second input state.

The EPIs in \cref{qEPI1,qEPI1/2,qEPI2} for continuous-variable quantum systems were proved using methods analogous to those used in proving
the classical EPIs \eqref{cEPI1} and \eqref{cEPI2}, albeit with suitable adaptations to the quantum setting. In contrast, the proof of our EPIs relies on completely different tools, namely, spectral majorization and concavity of functions.

\section{Preliminaries} \label{sec:prelim}

Let $\cH \simeq \C^d$ be a finite-dimensional Hilbert space (i.e., a complex Euclidean space), let $\cL(\cH)$ denote the set of linear operators acting on $\cH$, and let $\cD(\cH)$ be the set of density operators or \emph{states} on $\cH$:
\begin{equation}
  \cD(\cH):= \set[\big]{\rho \in \cL(\cH) : \rho \geq 0, \, \Tr \rho = 1}.
\end{equation}
Moreover, let $\cU(\cH)$ be the set of unitary operators acting on $\cH$. We denote the identity operator on $\cH$ by $I$. A quantum channel (or quantum operation) is given by a linear, completely positive, trace-preserving (CPTP) map $\mathcal{N}:\cL(\cH) \to \cL(\cK)$, with $\cH$ and $\cK$ being the input and output Hilbert spaces of the channel. For a state $\rho \in \D{d}$ with eigenvalues $\lambda_1, \dotsc, \lambda_d$, the \emph{von~Neumann entropy} $H(\rho)$ is equal to the Shannon entropy of the probability distribution $\set{\lambda_1, \dotsc, \lambda_d}$, i.e., $H(\rho) := - \Tr (\rho \log \rho) = - \sum_{i=1}^d \lambda_i \log \lambda_i$, where we take the logarithms to base $e$.

The proof of the quantum EPIs that we propose, relies on the concept of majorization~(see e.g.~\cite{Bhatia}). For convenience we recall its definition below, making use of the following notation: for any vector $\vec{u} = (u_1, u_2, \dotsc, u_d) \in \R^d$ let $u_1^\dar\geq u_2^\dar \geq \ldots \geq u_d^\dar$ denote the components of $\vec{u}$ arranged in non-increasing order.

\begin{definition}[Majorization]\label{def:Majorization}
For $\vec{u}, \vec{v} \in \R^d$, we say that $\vec{u}$ is \emph{majorised by} $\vec{v}$ and write
$\vec{u}\prec \vec{v}$ if
\begin{equation}
  \sum_{i=1}^k u_i^\dar \leq \sum_{i=1}^k v_i^\dar,
  \qquad
  \forall \, k \in \set{1, \dotsc, d}
\end{equation}
with equality at $k = d$.
\end{definition}

\begin{definition}
A function $f: \R^d \to \R$ is called \emph{Schur-concave}~\cite{Bhatia} if $f(\vec{u})\ge f(\vec{v})$ whenever $\vec{u} \prec \vec{v}$.
\end{definition}

The notion of majorization can be extended to quantum states as follows. For $\rho, \sigma \in \D{d}$, we write $\rho \prec \sigma$ if $\lambda(\rho) \prec \lambda(\sigma)$, where we use the notation $\lambda(\rho)$ to denote the vector of eigenvalues of $\rho$, arranged in non-increasing order: $\lambda(\rho) := \of[\big]{\lambda_1(\rho), \lambda_2(\rho), \dotsc, \lambda_d(\rho)}$ with
\begin{equation}
  \lambda_1(\rho) \geq
  \lambda_2(\rho) \geq \dotsb \geq
  \lambda_d(\rho).
\end{equation}

The following class of functions plays an important role in our paper. A canonical example of a function in this class is the von~Neumann entropy of a density matrix.

\begin{definition}\label{def:F}
Let $\F$ denote the class of functions $f : \D{d} \to \R$ satisfying the following properties:
\begin{enumerate}
  \item \emph{Concavity}: for any pair of states $\rho, \sigma \in \D{d}$ and $\forall \, \a \in [0,1]$:
    \begin{equation}
      f\of[\big]{\a \rho + (1-\a) \sigma} \geq \a f(\rho) + (1-\a) f(\sigma).
    \end{equation}
  \item \emph{Symmetry}: $f(\rho)$ depends only on the eigenvalues of $\rho$ and is symmetric in them; that is, there exists a symmetric (i.e.\ permutation-invariant) function $\phi_f:\R^d \to \R$ such that $f(\rho) = \phi_f(\lambda(\rho))$.
\end{enumerate}
\end{definition}

By restricting to diagonal states, it follows immediately that for every $f\in\F$ the corresponding function $\phi_f$ is concave. In turn, this means that $\phi_f$ is also Schur-concave~\cite[Theorem II.3.3]{Bhatia}.

\section{Main results}

We formulate a finite-dimensional version of the quantum addition rule given by \cref{eq:q-add}, which was introduced by K\"onig and Smith~\cite{KS13, KS14} in the context of continuous-variable quantum systems. Our operation, which we also denote by $\boxplus_\a$, is parameterized by $a \in [0,1]$. It combines a pair of $d$-dimensional quantum states $\rho$ and $\sigma$ according to the following quantum addition rule:
\begin{equation}
  \rho \boxplus_\a \sigma
  := \a \rho + (1-\a) \sigma - \sqrt{\a(1-\a)} \, i [\rho, \sigma],
  \label{eq:Box short}
\end{equation}
where $[\rho, \sigma] := \rho \sigma - \sigma \rho$. Note that if $[\rho, \sigma] = 0$ then $\rho \boxplus_\a \sigma$ is simply a convex combination of $\rho$ and $\sigma$. In \cref{sec:Beamsplitter} we prove that $\rho \boxplus_\a \sigma = \E_\a(\rho_1 \x \rho_2)$ for some quantum channel $\E_\a : \cD(\C^d \otimes \C^d) \mapsto \D{d}$, see \cref{eq:Unitary,eq:Box}, implying that $\rho \boxplus_\a \sigma$ is a valid state of a qudit. The main motivation behind introducing the map $\boxplus_\a$ is that, similar to its analogues (\cref{eq:c-add} and \cref{eq:q-add}) in the continuous-variable classical and quantum settings, it results in an interpolation between the two states which it combines, as the parameter $a$ is changed from $1$ to $0$.

We are now ready to summarize our main results, which are given by the following two theorems and corollary.

\begin{restatable}{theorem}{Main}\label{thm:Main}
For any $f \in \F$ (see \cref{def:F}), density matrices $\rho, \sigma \in \D{d}$, and any $\a \in [0,1]$,
\begin{equation*}
  f(\rho \boxplus_\a \sigma) \geq \a f(\rho) + (1-\a) f(\sigma).
\end{equation*}
\end{restatable}

\noindent
Note that from \cref{eq:Box short} it follows that for commuting states (and hence for diagonal states representing probability distributions) this inequality is equivalent to concavity of the function $f$. An extension of \cref{thm:Main} to three states is conjectured in~\cite{Ozols15}.

In analogy with the entropy power of p.d.f.s defined in \cref{eq:pdf ep}, as well as the entropy power and entropy photon number of continuous-variable quantum states, we use the von~Neumann entropy of finite-dimensional quantum systems to introduce similar quantities for qudits.

\begin{definition}\label{def:EP and N}
For any $\cs \geq 0$, we define the \emph{entropy power} $\EP_\cs$ and the \emph{entropy photon number} $N_\cs$ of $\rho \in \D{d}$ as follows:
\begin{align}
  \EP_\cs(\rho) &:= e^{\cs H(\rho)}, \label{def:EP} \\
  N_\cs(\rho) &:= g^{-1}(\cs H(\rho))
  \quad \text{where} \quad
  g(x) := (x+1) \log (x+1) - x \log x. \label{def:EPn}
\end{align}
\end{definition}
The function $g(x)$ behaves logarithmically, and is bounded from above and from below as
\begin{equation}
  1+\log(x+1/e) \le g(x) \le 1+\log(x+1/2),
\end{equation}
from which it follows that
\begin{equation}
  \exp(y-1) -1/2 \le g^{-1}(y)\le \exp(y-1)-1/e.
\end{equation}
Note that the quantity $N_\cs(\rho)$ does not have any obvious physical interpretation for qudits. It is simply defined in analogy to the continuous-variable quantum setting. Our motivation for looking at this quantity is that it allows us to prove a qudit analogue of the entropy photon number inequality (EPnI), which in the bosonic case remains an open problem.

Here we introduced the scaling parameter $\cs$ to account for the possibility of having a dependence on dimension or number of modes which is different from that arising in the continuous-variable classical and quantum settings. Recall that the classical EPI \eqref{cEPI2} for continuous random variables on $\R^d$ is stated in terms of $\EP_{2/d}$, while the quantum EPI \eqref{qEPI2} and the conjectured entropy photon number inequality for $m$-mode bosonic quantum states involves $\EP_{1/m}$ and $N_{1/m}$, respectively (see \cref{tab:EPIs}). Our next theorem establishes concavity of $\EP_\cs$ and $N_\cs$ for a wide range of values of $\cs$.

\begin{table}[t]
\centering
\newcommand{\yes}{$\checkmark$}
\newcommand{\mr}[1]{\multirow{2}{*}{#1}}
\newcommand{\mcr}[1]{\multicolumn{1}{|c|}{\multirow{1}{*}{#1}}}
\begin{tabular}{c|c|c|c|} \cline{2-4}
  & \multicolumn{2}{|c|}{Continuous} & Discrete         \\ \cline{2-4}
  & Classical        & Quantum       & Quantum          \\
  & ($m'$ dimensions) & ($m$ modes)   & ($d$ dimensions) \\ \hline
  \mcr{Entropy}      & \mr{\yes} & \mr{\yes} & \mr{\yes} \\
  \mcr{$H$}       &&&\\ \hline
  \mcr{Entropy}   &&&\\
  \mcr{power}                   & $c = 2/m'$ & $c = 1/m$ & $0 \leq \cs \leq 1/(\log d)^2$ \\
  \mcr{$\EP_\cs$} &&&\\ \hline
  \mcr{Entropy}   &&&\\
  \mcr{photon}                  & \mr{---}  & $c = 1/m$ & \mr{\!\!$0 \leq \cs \leq 1/(d-1)$} \\
  \mcr{number}    &&\small(conjectured)& \\
  \mcr{$N_\cs$}   &&&\\ \hline
\end{tabular}
\caption{\label{tab:EPIs}Summary of classical and quantum EPIs.}
\end{table}

\begin{theorem}\label{thm:Concavity}
For $\rho \in \D{d}$, the following functions are concave:
\begin{itemize}
  \item the entropy power $\EP_\cs(\rho)$ for $0 \leq \cs \leq 1/(\log d)^2$,
  \item the entropy photon number $N_\cs(\rho)$ for $0 \leq \cs \leq 1/(d-1)$.
\end{itemize}
\end{theorem}

Since $\EP_\cs(\rho)$ and $N_\cs(\rho)$ depend only on the eigenvalues of $\rho$ and are symmetric in them, the above theorem ensures that $\EP_\cs$ and $N_\cs$ belong to the class of functions $\F$ given in \cref{def:F}. From \cref{thm:Main,thm:Concavity}, and the concavity of the von Neumann entropy, we obtain the following.

\begin{corollary}\label{thm:EPIs}
For any pair of density matrices $\rho, \sigma \in \D{d}$ and any $\a \in [0,1]$,
\begin{align}
  H(\rho \boxplus_\a \sigma) &\geq \a H(\rho) + (1-\a) H(\sigma), \label{EPI H} \\
  e^{\cs H(\rho \boxplus_\a \sigma)} &\geq \a e^{\cs H(\rho)} + (1-\a) e^{\cs H(\sigma)} &&\quad \text{for} \quad 0 \leq \cs \leq 1/(\log d)^2, \label{EPI exp} \\
  N_\cs(\rho \boxplus_\a \sigma) &\geq \a N_\cs(\rho) + (1-\a) N_\cs(\sigma)             &&\quad \text{for} \quad 0 \leq \cs \leq 1/(d-1). \label{EPnI-d}
\end{align}
\end{corollary}
\noindent
Henceforth, we refer to \cref{EPI H,EPI exp} as qudit EPIs and \cref{EPnI-d} as qudit EPnI. A summary of values of the parameter $\cs$ for which classical and quantum EPIs hold is given in~\cref{tab:EPIs}.

In addition, \cref{thm:Main} also holds for the R\'enyi entropy $H_\alpha(\rho)$ of order $\alpha$~\cite{Renyi}, for $\alpha \in [0,1)$, the subentropy $Q(\rho)$~\cite{JRW94, DDJB14}, defined as follows:
\begin{align}
  H_\alpha(\rho) &:= \frac{1}{\alpha-1} \log \of{\Tr \rho^\alpha}, \label{def:alpha} \\
  Q(\rho) &:= -\sum_{i=1}^n \frac{\lambda_i^n}{\prod_{j \ne i}(\lambda_i - \lambda_j)} \log \lambda_i, \label{def:sub}
\end{align}
where $\lambda_1, \dotsc, \lambda_d$ denote the eigenvalues of $\rho$. If some eigenvalues coincide (or are zero), $Q(\rho)$ is defined to be the corresponding limit of the above expression, which is always well-defined and finite. The above functions are clearly symmetric in the eigenvalues of $\rho$ and are known to be concave. Hence, they belong to the class $\F$ and thus obey the inequality in \cref{thm:Main}.

\section{An addition rule for qudit states} \label{sec:Beamsplitter}

In this section we show how we arrive at the quantum addition rule for qudits, \eqref{eq:Box short}, for which we prove a family of EPIs. This rule is based on a continuous version of the swap operation and it mimics the behavior of a beamsplitter.

\subsection{Beamsplitter}

\begin{figure}
\centering
\input{fig-beamsplitter.tex}
\caption{A comparison of a beamsplitter and the partial swap operation.\label{fig:Beamsplitter}}
\end{figure}
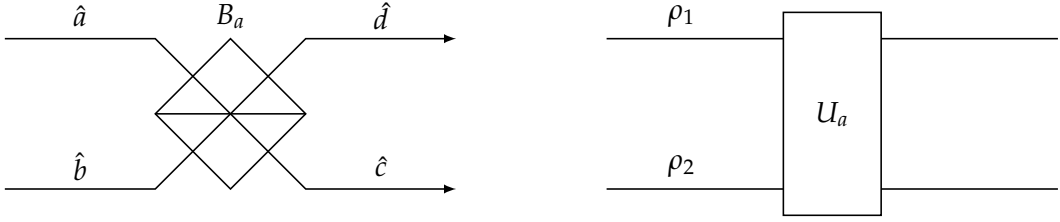

Let $\hat{a}, \hat{b}$ denote the annihilation operators of the two bosonic input modes of a beamsplitter and $\hat{c}, \hat{d}$ denote the annihilation operators of the two output modes (see \cref{fig:Beamsplitter}, left). Then the action of a beamsplitter on the input modes is described as follows~\cite{RevMod}:
\begin{equation}
  \mx{\hat{c} \\ \hat{d}} = B
  \mx{\hat{a} \\ \hat{b}},
\end{equation}
where $B$ is an arbitrary $2 \times 2$ unitary matrix also known as the \emph{scattering matrix}. In particular, let us choose
\begin{equation}
  B_\a := \mx{\sqrt{\a} & i\sqrt{1-\a} \\ i\sqrt{1-\a} & \sqrt{\a}},
\end{equation}
where $\a \in [0,1]$ is the \emph{transmissivity} of the beamsplitter (note that this choice slightly differs from the one corresponding to \cref{eq:cs}). As $\a$ changes from $1$ to $0$, $B_\a$ interpolates between the identity matrix $I$ and $i \sigma_x$ (where $\sigma_x$ is the Pauli-$x$ matrix). Indeed, we can write
\begin{equation}
  B_\a = \sqrt{\a} \, I + i \sqrt{1-\a} \, \sigma_x.
  \label{eq:Ba}
\end{equation}
In particular, up to an unimportant phase, $B_0$ acts as the swap operation $\sigma_x$ between the two modes. Thus, for intermediate values of $\a$, we can interpret $B_\a$ as an operator that partially swaps the two modes. Following this intuition, in the next section we introduce a partial swap operation for two qudits (see \cref{fig:Beamsplitter}, right) that mimics the action of $B_\a$. It is also described by a unitary matrix that is an interpolation between the identity and a swap operation (up to a phase factor), but with a swap that exchanges two qudits.

\subsection{The partial swap operator}

Let $\{\ket{i}\}_{i=1}^d$ denote the standard basis of $\cH \simeq \C^d$. Then $\{\ket{i,j}\}_{i,j=1}^d$ is an orthonormal basis of $\cH \otimes \cH$. The qudit \emph{swap operator} $S \in \cU(\cH \otimes \cH)$ is defined through its action on the basis vectors $\ket{i,j}$ as follows:
\begin{equation}
  S \ket{i,j} := \ket{j,i} \quad \text{for all } i,j \in \set{1, \dotsc, d}
\end{equation}
and can be expressed as
\begin{equation}
  S = \sum_{i,j=1}^d \ket{i}\bra{j}\otimes \ket{j}\bra{i}.
\end{equation}
Clearly, $S\ct = S$ and $S^2 = I$. In analogy with the beamsplitter scattering matrix \cref{eq:Ba}, we define a qudit partial swap operator as a unitary interpolation between the identity and the swap operator.

Since $S$ is Hermitian, we can view it as a Hamiltonian. The evolution for time $t \in \R$ under its action is given by the following unitary operator, where we used the fact that $S^2 = I$:
\begin{equation}
  \exp(i t S)
  = \sum_{n=0}^\infty \frac{(it)^n}{n!} S^n
  = I \cos t + i S \sin t. \label{eq:cis}
\end{equation}
In particular, $\exp(i (\pi/2) S)=iS$, so $\exp(itS)=(iS)^{2t/\pi}$.
Thus, as $t$ changes from $0$ to $\pi/2$, this unitary operator interpolates between $I$ and $iS$, the swap gate up to a global phase. (We are interested only in how this matrix acts under conjugation, so the global phase can be ignored.) We reparametrize \cref{eq:cis} by $(\sqrt{\a}, \sqrt{1-\a}) = (\cos t, \sin t)$ and refer to the resulting unitary as the partial swap operator.

\begin{definition}\label{def:Ua}
For $\a \in [0,1]$, the \emph{partial swap operator} $U_\a \in \cU(\C^d \x \C^d)$ is the unitary operator
\begin{equation}
  U_\a := \sqrt{\a} \, I + i \sqrt{1-\a} \, S.
  \label{eq:Ua}
\end{equation}
\end{definition}

\noindent
Up to the sign of $i$, any complex linear combination of $I$ and $S$ that is unitary is of this form~\cite{Ozols15}.
Note that $U_1 = I$ while $U_0 = iS$ acts as the qudit swap under conjugation: $U_0 (\rho_1 \x \rho_2) U_0\ct = \rho_2 \x \rho_1$.

\begin{example}[Qubit case: $d=2$]
The matrix representation of the partial swap operator for qubits is
\begin{equation}
  U_\a = \mx{
           \sqrt{\a} + i \sqrt{1-\a} & 0 & 0 & 0 \\
           0 & \sqrt{a} & i \sqrt{1-\a} & 0 \\
           0 & i \sqrt{1-\a} & \sqrt{a} & 0 \\
           0 & 0 & 0 & \sqrt{\a} + i \sqrt{1-\a}
         }.
\end{equation}
\end{example}

\subsection{The partial swap channel}

Consider a family of CPTP maps $\E_\a: \DD{d} \to \D{d}$ parameterized by $\a \in [0,1]$ and defined in terms of the partial swap operator $U_\a$ given in \cref{eq:Ua}. For any $\rho_{12} \in \cD(\cH_1 \x \cH_2)$ with $\cH_1, \cH_2 \simeq \C^d$, let
\begin{equation}
  \E_\a(\rho_{12}) := \Tr_2 \of[\big]{U_a \rho_{12} U_a\ct},
  \label{eq:E}
\end{equation}
where we trace out the second system. We are particularly interested in the case in which the input state $\rho_{12}$ is a product state, i.e., $\rho_{12} = \rho_1 \x \rho_2$ for some $\rho_1, \rho_2 \in \D{d}$. When $\E_\a$ is applied on such states, it combines the two density matrices $\rho_1$ and $\rho_2$ in a non-trivial manner, which mimics the action of a beamsplitter~\cite{KS13,KS14}. To wit, $\E_0(\rho_1 \x \rho_2) = \rho_2$ and $\E_1(\rho_1 \x \rho_2) = \rho_1$, while for general $\a \in [0,1]$ the output of $\E_\a(\rho_1 \x \rho_2)$ continuously interpolates between $\rho_1$ and $\rho_2$. The following lemma provides an explicit expression for the resulting state (this expression has independently appeared also in \cite{LMR14} in the context of quantum algorithms).

\begin{lemma}\label{lem:Swap}
Let $\E_\a$ denote the map defined in \cref{eq:E} and $[\rho_1, \rho_2] := \rho_1 \rho_2 - \rho_2 \rho_1$. Then for $\rho_1, \rho_2 \in \D{d}$,
\begin{equation}
  \E_\a(\rho_1 \x \rho_2)
  = a \rho_1 + (1-a) \rho_2 - \sqrt{a(1-a)} \, i[\rho_1, \rho_2]. \label{eq:E action}
\end{equation}
\end{lemma}

\begin{remark}
When $(\sqrt{\a}, \sqrt{1-\a}) = (\cos t, \sin t)$ for some $t \in [0, \pi/2]$, this is an elliptic path in $\D{d}$:
\begin{equation}
  \E_\a(\rho_1 \x \rho_2)
  =  \frac{\rho_1 + \rho_2}{2} + \cos 2t \cdot \frac{\rho_1-\rho_2}{2} - \sin 2t \cdot \frac{i}{2}[\rho_1,\rho_2].
  \label{eq:E action_2t}
\end{equation}
If we flip the sign of $i$ or allow $t \in [-\pi/2, 0]$, we get the other half of the ellipse (see also~\cite{Ozols15}).
\end{remark}

\begin{proof}[Proof of \cref{lem:Swap}]
Using \cref{eq:Ua} we get
\begin{align}
  U_\a (\rho_1 \x \rho_2) U_\a\ct
 &= \of[\big]{\sqrt{\a} \, I + i \sqrt{1-\a} \, S} (\rho_1 \x \rho_2)
    \of[\big]{\sqrt{\a} \, I - i \sqrt{1-\a} \, S} \nonumber \\
 &= a \, \rho_1 \x \rho_2
  + (1-a) \, \rho_2 \x \rho_1
  + i \sqrt{a(1-a)} \, \of[\big]{ S (\rho_1 \x \rho_2) - (\rho_1 \x \rho_2) S }.
\end{align}
After tracing out the second system, the first two terms of the above expression give the first two terms of \cref{eq:E action}. To get the last term of \cref{eq:E action}, note that
\begin{align}
  \Tr_2 \of[\big]{(\rho_1 \x \rho_2) S}
  &= \sum_{k=1}^d (I \x \bra{k}) \of[\bigg]{ (\rho_1 \x \rho_2) \sum_{i,j=1}^d \ket{i} \bra{j} \x \ket{j} \bra{i}} (I \x \ket{k}) \nonumber \\
  &= \sum_{i,j,k=1}^d \rho_1 \ket{i} \bra{j} \x \bra{k} \rho_2 \ket{j} \braket{i}{k} \nonumber \\
  &= \rho_1 \sum_{i,j=1}^d  \bra{i} \rho_2 \ket{j} \ket{i} \bra{j} = \rho_1 \rho_2
\end{align}
and similarly $\Tr_2 \of[\big]{S (\rho_1 \x \rho_2)} = \rho_2 \rho_1$.
Hence,
\begin{equation}
  \Tr_2 \of[\big]{S (\rho_1 \x \rho_2) - (\rho_1 \x \rho_2) S}
  = \rho_2 \rho_1 - \rho_1 \rho_2
  = [\rho_2, \rho_1],
\end{equation}
which yields the last term of \cref{eq:E action}.
\end{proof}

One can check that the action of the channel $\E_\a$ on an arbitrary state $\rho \in \DD{d}$ (i.e., not necessarily a product state) can be expressed as $\E_\a(\rho) = \sum_{k=1}^d A_k \rho A_k\ct$ with the Kraus operators $A_k$ given by
\begin{equation}
  A_k := \sqrt{\a} \, I \x \bra{k} + i \sqrt{1-\a} \, \bra{k} \x I
  \quad \text{for} \quad
  k \in \set{1, \dotsc, d}.
\end{equation}
Using \cref{lem:Swap}, we introduce a qudit addition rule which combines two $d \times d$ density matrices.

\begin{definition}[Qudit addition rule]
For any $\a \in [0,1]$ and any $\rho_1, \rho_2 \in \D{d}$, we define
\begin{align}
  \rho_1 \boxplus_\a \rho_2
  &:= \E_\a(\rho_1 \x \rho_2)
    = \Tr_2 \of[\big]{ U_\a (\rho_1 \x \rho_2) U_\a\ct } \label{eq:Unitary} \\
   &= \a \rho_1 + (1-\a) \rho_2 - \sqrt{\a(1-\a)} \, i [\rho_1, \rho_2]. \label{eq:Box}
\end{align}
\end{definition}

\noindent
This operation is bilinear under convex combinations and obeys $\rho_1 \boxplus_0 \rho_2 = \rho_2$ and $\rho_1 \boxplus_1 \rho_2 = \rho_1$. A generalization of \cref{eq:Box} to three states is given in~\cite{Ozols15}.

\begin{example}[Qubit case: $d=2$]
Let $\vec{r}$, $\vec{r}_1$, $\vec{r}_2$ denote the Bloch vectors (see \cref{apx:Bloch sphere}) of states $\rho_1 \boxplus_\a \rho_2$, $\rho_1$, $\rho_2$, respectively. Using the properties of Pauli matrices, one can show that \cref{eq:Box} is equivalent to
\begin{equation}
  \vec{r}
  = \a \vec{r}_1
  + (1-\a) \vec{r}_2
  + \sqrt{\a(1-\a)} \, \vec{r}_1 \times \vec{r}_2,
  \label{eq:r}
\end{equation}
where $\vec{r}_1 \times \vec{r}_2$ denotes the cross product of $\vec{r}_1$ and $\vec{r}_2$.
\end{example}

\subsection{Partial swap vs. mixing}

Are there any other natural operations $\widetilde{\boxplus}_\a$ for combining two states for which the EPIs that we prove also hold? A trivial example is the CPTP map $\tE_a$ that acts on product states by mixing the two factors, i.e. for which
\begin{equation}\label{one}
  \tE_a (\rho \otimes \sigma) := \rho \, \widetilde{\boxplus}_a \, \sigma = a \rho + (1-a) \sigma.
\end{equation}
It has the following $2d$ Kraus operators: $A_k := \sqrt{a} \, I \x \bra{k}$ and $B_k := \bra{k} \x \sqrt{1-a} \, I$ for $k \in \set{1,\dotsc,d}$, and requires an ancillary qubit. Note, however, that for this choice of $\tE_a$ (and hence $\widetilde{\boxplus}_a$) \cref{thm:Main} is trivial as it simply restates the concavity of the function $f$.

In contrast, the partial swap channel $\E_a$ has the following features: (i) it yields non-trivial EPIs (that are not simply a statement of concavity), and (ii) it does not require an ancillary qubit, so it has only $d$ Kraus operators, the minimal number required for tracing out a $d$-dimensional system.

\section{Proof of \cref{thm:Main}}

In this section we prove \cref{thm:Main}, our main result. Due to the very different setup as compared to the work of K\"onig and Smith, with our addition rule acting at the level of states rather than at the level of field operators, our mathematical treatment is entirely different from theirs and bears no obvious similarity with the classical case either. Instead of proceeding via quantum generalizations of Young's inequality, Fisher information and de Bruijn's identity, the main ingredient in our proof is the following majorization relation relating the spectrum of the output state to the spectra of the input states.

\begin{theorem}\label{thm:maj}
For any pair of density matrices $\rho, \sigma \in \D{d}$ and any $\a \in [0,1]$,
\begin{equation}
  \lambda(\rho \boxplus_\a \sigma) \prec \a \lambda(\rho) + (1-\a) \lambda(\sigma).
  \label{eq:Majorization}
\end{equation}
\end{theorem}

\begin{remark}
For fields corresponding to the action of a beamsplitter, the addition rule translates to linearly combining the covariance matrices $\gamma$~\cite{KS14}:
\begin{equation}
  \gamma(\rho \boxplus_a \sigma) = a \gamma(\rho) + (1-a) \gamma(\sigma).
\end{equation}
When the incoming quantum fields are both Gaussian, an inequality closely related to \cref{eq:Majorization} holds. Denoting by $\nu(A)$ the symplectic eigenvalues of a covariance matrix $A$, Hiroshima \cite{Hiroshima} has shown that for any $A,B\ge0$,
\begin{equation}
\nu(A+B)\prec^w \nu(A)+\nu(B),
\end{equation}
where $\prec^w$ stands for \emph{weak supermajorization} \cite{Bhatia}. Applied to $\gamma(\rho)$ and $\gamma(\sigma)$, this inequality can be used to derive an EPI for Gaussian fields in a similar way as we have done for qudits.
\end{remark}

We will first show how our main result follows from \cref{thm:maj}, as this is straightforward, and then proceed with the proof of the latter, which is the bulk of the work. We restate \cref{thm:Main} here, for convenience.

\Main*

\begin{proof}
Assume \cref{thm:maj} has been established. Let $\trho, \tsigma \in \D{d}$ be diagonal states whose entries are the eigenvalues of $\rho$ and $\sigma$ (respectively), arranged in non-increasing order. Since $\lambda(\trho) = \lambda(\rho)$ and $\lambda(\tsigma) = \lambda(\sigma)$, \cref{eq:Majorization} can be equivalently written as
\begin{align}
  \lambda(\rho \boxplus_\a \sigma)
  &\prec \a \lambda(\trho) + (1-\a) \lambda(\tsigma), \nonumber \\
  &= \lambda\of[\big]{\a \trho + (1-\a) \tsigma}. \label{eq:prec}
\end{align}
For any function $f \in \F$ (see \cref{def:F}) \cref{eq:prec} implies that
\begin{align}
  f(\rho \boxplus_\a \sigma)
  &\geq f\of[\big]{\a \trho + (1-\a) \tsigma} \nonumber \\
  &\geq \a f(\trho) + (1-\a) f(\tsigma) \nonumber \\
  & = \a f(\rho) + (1-\a) f(\sigma),
\end{align}
where the first inequality follows by Schur-concavity, the second inequality follows from concavity, and the last line follows by symmetry. Thus, we have arrived at the statement of \cref{thm:Main}.
\end{proof}

It remains to prove \cref{thm:maj}. For this we will need the following two lemmas.

\begin{lemma}[von Neumann \protect{\cite[p.~55]{vonNeumann}}]\label{lem:vonN}
Let $\mc{L}$ and $\mc{M}$ be two subspaces of a vector space and let $P(\mc{L})$ and $P(\mc{M})$ denote the corresponding projectors. Then
\begin{equation}
  P(\mc{L} \cap \mc{M}) = \lim_{n \to \infty} \of[\big]{P(\mc{L}) P(\mc{M})}^n.
\end{equation}
\end{lemma}

\begin{lemma}\label{lem:Min}
For $0 \le x,y \le 1$, the following inequality holds:
\begin{equation}
  x y + \sqrt{x (1-x) y (1-y)} \geq \min \set{x,y}.
\end{equation}
\end{lemma}

\begin{proof}
Without loss of generality, we can assume that $0 \le x \le y \le 1$, so we need to show that
\begin{equation}
  x \le x y + \sqrt{x (1-x) y (1-y)}.
\end{equation}
Since $x \le y$, we have $x - x y \le y - y x$, or $x (1-y) \le y (1-x)$. By the above assumption, each side is non-negative. Taking the geometric mean of each side with $x(1-y)$ then yields
\begin{equation}
  x (1-y) \le \sqrt{x (1-y) y (1-x)},
\end{equation}
which is equivalent to what we had to prove.
\end{proof}

Now we are ready to prove \cref{thm:maj}. (Note that subsequently our proof has been simplified by Carlen, Lieb, and Loss~\cite{CLL16}.)

\begin{proof}[Proof of \cref{thm:maj}]
The expression $\rho \boxplus_\a \sigma = \a \rho + (1-\a) \sigma - \sqrt{\a(1-\a)} \, i[\rho,\sigma]$ can be written as follows:
\begin{equation}
  \rho \boxplus_\a \sigma
  = \a(\rho - \rho^2)
  + (1-\a) (\sigma - \sigma^2)
  + (\sqrt{\a}\,\rho + i\sqrt{1-\a}\,\sigma)
    (\sqrt{\a}\,\rho + i\sqrt{1-\a}\,\sigma)\ct.
\end{equation}
It is convenient to express the state $\rho \boxplus_\a \sigma$ as $T T\ct$ for some $1 \times 3$ block-matrix
\begin{equation}
  T = \left(T_1  \quad T_2 \quad T_3 \right).
\end{equation}
We choose $T := A + iB$ where $A$ and $B$ are the following $1 \times 3$ block matrices:
\begin{IEEEeqnarray}{rCr.lc?c?c}
  A &:=& \sqrt{\a}   & \of[\Big]{& (\rho-\rho^2)^{1/2} & 0 & \rho },       \label{eq:A} \\
  B &:=& \sqrt{1-\a} & \of[\Big]{& 0 & (\sigma-\sigma^2)^{1/2} & \sigma }. \label{eq:B}
\end{IEEEeqnarray}
Here the operator square roots are well-defined, since $X \geq X^2$ for any matrix $I \geq X \geq 0$. Also, note that all blocks of $A$ and $B$ (and hence of $T$) are Hermitian. One can easily check that
\begin{align}
  AA\ct &= \a (\rho-\rho^2) + \a \rho^2 = \a \rho, \label{eq:AA} \\
  BB\ct &= (1-\a) (\sigma-\sigma^2) + (1-\a) \sigma^2 = (1-\a) \sigma, \\
  TT\ct &= (A + iB)(A\ct - iB\ct) \nonumber\\
        &= AA\ct + BB\ct - i(AB\ct - BA\ct) \nonumber\\
        &= \a \rho + (1-\a) \sigma - i\sqrt{\a(1-\a)} [\rho, \sigma] \nonumber\\
        &= \rho \boxplus_\a \sigma.
\end{align}

Given these expressions, we can rewrite \cref{eq:Majorization} as
\begin{equation}
  \lambda(TT\ct) \prec \lambda(AA\ct) + \lambda(BB\ct).
  \label{eq:TAB}
\end{equation}
If $A$ and $B$ had been positive semidefinite, this inequality would have followed straight-away from Theorem 3.29 in \cite{Zhan}. Nevertheless, we can adapt the proof of this theorem to our needs. Note that
\begin{equation}
  \Tr(TT\ct) = \Tr(AA\ct) + \Tr(BB\ct) - i \Tr[A,B] = \Tr(AA\ct) + \Tr(BB\ct),
\end{equation}
since $\Tr[A,B]=0$ by the cyclicity of the trace. Hence,
\begin{equation}
  \sum_{j=1}^d \lambda_j(TT\ct) = \sum_{j=1}^d \lambda_j(AA\ct) + \sum_{j=1}^d \lambda_j(BB\ct).
\end{equation}
From this and \cref{def:Majorization} we see that \cref{eq:TAB} is equivalent to
\begin{equation}
  \sum_{j=d-k+1}^d \lambda_j(TT\ct) \ge
  \sum_{j=d-k+1}^d \lambda_j(AA\ct) +
  \sum_{j=d-k+1}^d \lambda_j(BB\ct),
  \qquad
  \forall \, k \in \set{1, \dotsc, d}.
  \label{eq:Majorization sums}
\end{equation}
The left-hand side of the above inequality can be expressed variationally as follows (see e.g.~Corollary 4.3.39 in \cite{HornJohnson}):
\begin{equation}
  \sum_{j=d-k+1}^d \lambda_j(TT\ct) = \min \set[\big]{\Tr (U_k\ct TT\ct U_k) : U_k \in M_{d,k}, U_k\ct U_k = I_k},
\end{equation}
where $ M_{d,k}$ denotes the set of $d \times k$ matrices, and  $I_k \in M_{k,k}$ is the identity matrix. Note that the constraint $U_k\ct U_k = I_k$ is equivalent to $U_k$ being a $d \times k$ matrix consisting of $k$ columns of a $d \times d$ unitary matrix $U$. We can express $U_k$ as $U_k = U I_{k,d}$, where $I_{k,d} := I_k \oplus 0_{d-k}$, with $0_k \in M_{d-k,d-k}$ being a matrix with all entries equal to zero. Hence, $U_k U_k\ct = U I_{k,d} U\ct$, which is a projector of rank $k$. Clearly, $U_k U_k\ct \leq I_d$, so that
\begin{align}
  \Tr (U_k\ct TT\ct U_k)
  &=   \sum_{l=1}^3 \Tr (U_k\ct T_l T_l\ct U_k) \nonumber\\
  &\ge \sum_{l=1}^3 \Tr (U_k\ct T_l U_k\; U_k\ct T_l\ct U_k) \nonumber\\
  &=   \sum_{l=1}^3 \Tr \of[\big]{ U_k\ct (A_l + i B_l) U_k\; U_k\ct (A_l\ct - i B_l\ct) U_k } \nonumber\\
  &=   \sum_{l=1}^3 \Tr (U_k\ct A_l U_k)^2
     + \sum_{l=1}^3 \Tr (U_k\ct B_l U_k)^2
    -i \sum_{l=1}^3 \Tr [U_k\ct A_l U_k, U_k\ct B_l U_k] \nonumber\\
&=   \sum_{l=1}^3 \Tr (U_k\ct A_l U_k)^2
     + \sum_{l=1}^3 \Tr (U_k\ct B_l U_k)^2\label{eq:Last}
\end{align}
where we used that $A_l\ct = A_l$ and $B_l\ct = B_l$ for all $l$.

To complete the proof of \cref{eq:Majorization sums}, we will show that $\sum_{l=1}^3 \Tr (U_k\ct A_l U_k)^2 \geq \sum_{j=d-k+1}^d \lambda_j(AA\ct)$, with a corresponding inequality for $B$ following in the same way. From the definition of $A$ we have
\begin{equation}
  \sum_{l=1}^3 \Tr (U_k\ct A_l U_k)^2
  = \a \of[\Big]{ \Tr \of[\big]{ U_k\ct (\rho-\rho^2)^{1/2} U_k }^2 + \Tr(U_k\ct \rho U_k)^2 }.
\end{equation}
Recall from \cref{eq:AA} that $AA\ct = \a \rho$. Therefore, we have to show that
\begin{equation}
  \Tr \of[\big]{ U_k\ct (\rho-\rho^2)^{1/2} U_k }^2 + \Tr(U_k\ct \rho U_k)^2 \geq \sum_{j=d-k+1}^d \lambda_j(\rho),
  \qquad
  \forall \, k \in \set{1, \dotsc, d}.
  \label{eq:A majorization}
\end{equation}
Let $\rho = \sum_{i=1}^d \lambda_i \proj{\psi_i}$ be the eigenvalue decomposition of $\rho$, with the eigenvalues $\lambda_i$ being arranged in non-increasing order:
\begin{equation}
  \lambda_1 \geq \lambda_2 \geq \dotsb \geq \lambda_d.
  \label{eq:Sorted lambdas}
\end{equation}
Then the right-hand side of \cref{eq:A majorization} is $\sum_{j=d-k+1}^d \lambda_j$ while the left-hand side is
\begin{equation}
  \Tr \of*{ \sum_{i=1}^d \sqrt{\lambda_i (1-\lambda_i)} \, U_k\ct \proj{\psi_i} U_k }^2
+ \Tr \of*{ \sum_{i=1}^d \lambda_i \, U_k\ct \proj{\psi_i} U_k }^2.
\end{equation}
Expanding the squares gives
\begin{equation}
  \sum_{i,j=1}^d \of*{ \sqrt{\lambda_i (1-\lambda_i) \lambda_j (1-\lambda_j)} + \lambda_i \lambda_j }
  \Tr \of*{ U_k\ct \proj{\psi_i} U_k U_k\ct \proj{\psi_j} U_k }.
  \label{eq:roots}
\end{equation}
Noting that
\begin{equation}
  C_{ij}
  := \Tr \of*{ U_k\ct \proj{\psi_i} U_k U_k\ct \proj{\psi_j} U_k }
  = \abs[\big]{ \bra{\psi_i} U_k U_k\ct \ket{\psi_j} }^2
  \label{eq:Pij}
\end{equation}
is a non-negative real quantity, we can use \cref{lem:Min} to show that the expression \eqref{eq:roots}, and hence the left-hand side of \cref{eq:A majorization}, is bounded below by
\begin{equation}
  \sum_{i,j=1}^d \min \set{\lambda_i,\lambda_j} \, C_{ij}.
\end{equation}

Let $\Lambda$ be the matrix whose elements are $\Lambda_{ij} := \min \set{\lambda_i,\lambda_j}$:
\begin{equation}
  \Lambda = \mx{
    \lambda_1 & \lambda_2 & \lambda_3 & \cdots & \lambda_d \\
    \lambda_2 & \lambda_2 & \lambda_3 & \cdots & \lambda_d \\
    \lambda_3 & \lambda_3 & \lambda_3 & \cdots & \lambda_d \\
    \vdots    & \vdots    & \vdots    & \ddots & \vdots    \\
    \lambda_d & \lambda_d & \lambda_d & \cdots & \lambda_d
  }.
\end{equation}
For $m \in \set{1, \dotsc, d}$, we define matrices $E_m$ of size $d \times d$ such that
\begin{equation}
  (E_m)_{ij}
  := \begin{cases}
      1, & \text{for } 1 \le i,j \le m \\
      0, & \text{otherwise}.
    \end{cases}
\end{equation}
Then we can write
\begin{equation}
  \Lambda = \lambda_d E_d + \sum_{m=1}^{d-1} (\lambda_m - \lambda_{m+1}) E_m.
\end{equation}
Hence,
\begin{equation}
  \sum_{i,j=1}^d \min \set{\lambda_i,\lambda_j} \, C_{ij}
  \equiv \sum_{i,j=1}^d  \Lambda_{ij} C_{ij}= \lambda_d \sum_{i,j=1}^d (E_d \circ C)_{ij}
  + \sum_{m=1}^{d-1} (\lambda_m - \lambda_{m+1}) \sum_{i,j=1}^d (E_m \circ C)_{ij},
  \label{eq:Differences}
\end{equation}
where we use the notation $(A \circ B)_{ij} := A_{ij} B_{ij}$ for $d \times d$ matrices $A$ and $B$.

If we define $\pi(m) := \sum_{i,j=1}^d (E_m \circ C)_{ij} = \sum_{i,j=1}^m C_{ij}$, we can write \cref{eq:Differences} as
\begin{equation}
  \sum_{i,j=1}^d \min \set{\lambda_i,\lambda_j} \, C_{ij}
  = \lambda_d \pi(d) + \sum_{m=1}^{d-1} (\lambda_m - \lambda_{m+1}) \pi(m).
  \label{eq:Arcs}
\end{equation}
Recall from \cref{eq:Sorted lambdas} that the eigenvalues $\lambda_i$ are arranged in non-increasing order, so all coefficients $\lambda_d$ and $\lambda_m - \lambda_{m+1}$ are non-negative, so it only remains to find a lower bound on $\pi(m)$.

Recall from \cref{eq:Pij} that
\begin{align}
  \sum_{i,j=1}^m C_{ij}
  &= \sum_{i,j=1}^m \Tr \of*{ U_k\ct \proj{\psi_i} U_k U_k\ct \proj{\psi_j} U_k } \nonumber \\
  &= \Tr \of[\big]{ U_k\ct Q_m U_k U_k\ct Q_m U_k } \nonumber \\
  &= \Tr \of      { P_k Q_m }^2,
\end{align}
where $P_k := U_k U_k\ct$ and $Q_m := \sum_{i=1}^m \proj{\psi_i}$ are rank-$k$ and rank-$m$ projectors, respectively. Note that $\Tr \of{ P_k Q_m }^n$ is monotonically decreasing as a function of $n \in \N$, so
\begin{equation}
  \Tr \of{ P_k Q_m }^2
  \geq \lim_{n \to \infty} \Tr \of{ P_k Q_m }^n
     = \Tr \lim_{n \to \infty} \of{ P_k Q_m }^n
     = \Tr R
\end{equation}
where $R := \lim_{n \to \infty} \of{ P_k Q_m }^n$. If $\mc{S}_k$ and $\mc{S}_m$ are the subspaces of $\C^d$ corresponding to projectors $P_k$ and $Q_m$ respectively, then, by \cref{lem:vonN}, $R$ is the projector onto $\mc{S}_k \cap \mc{S}_m$. Since $\dim \mc{S}_k = k$ and $\dim \mc{S}_m = m$, we get
\begin{equation}
  \Tr R
  = \dim (\mc{S}_k \cap \mc{S}_m)
  \geq \max \set{0, k + m - d}.
\end{equation}

Putting everything together, we obtain
\begin{equation}
  \pi(m) = \sum_{i,j=1}^m C_{ij} \geq \max \set{0, k + m - d}.
\end{equation}
When we substitute this in \cref{eq:Arcs}, we get
\begin{align}
  \sum_{i,j=1}^d \min \set{\lambda_i,\lambda_j} \, C_{ij}
  \geq \lambda_d k + \sum_{m=d-k+1}^{d-1} (\lambda_m-\lambda_{m+1}) (k + m - d).
\end{align}
The right-hand side of the above inequality is simply equal to
\begin{equation}
      (\lambda_{d-k+1} - \lambda_{d-k+2})
  + 2 (\lambda_{d-k+2} - \lambda_{d-k+3})
  + \dotsb
  + (k-1) (\lambda_{d-1} - \lambda_{d})
  + k \lambda_d
  = \sum_{j=d-k+1}^d \lambda_j,
\end{equation}
which proves \cref{eq:A majorization} and therefore the theorem.
\end{proof}

\section{Concavity of entropy power and entropy photon number} \label{sec:Concavity}

In this section we prove \cref{thm:Concavity}, which establishes concavity of the entropy power $\EP_c$ and the entropy photon number $N_c$ for qudits (see \cref{def:EP and N}). Note that both $\EP_c(\rho)$ and $N_c(\rho)$ are twice-differentiable and monotonously increasing functions of the von~Neumann entropy $H(\rho)$. Hence, our strategy for establishing \cref{thm:Concavity} is to solve the following more general problem.

\begin{problem}
Let $h: \R \to \R$ be any twice-differentiable and monotonously increasing function. For which values of $\cs \geq 0$ is $f_\cs(\rho) := h(\cs H(\rho))$ concave on the set of $d$-dimensional quantum states?
\end{problem}

Since $H(\rho)$ is already concave, the function $f_\cs(\rho) = h(\cs H(\rho))$ is guaranteed to be concave for any $\cs \geq 0$ whenever $h$ is monotonously increasing \emph{and} concave. However, there are many more functions $h$ which are not necessarily be concave---in fact, they could even be convex---yet produce a concave function $f_\cs$ for a limited range of constants $\cs$. Our goal is to obtain a condition on pairs $(h,\cs)$ under which the function $f_c$ is concave.

To prove the concavity of $f_c$ on $\D{d}$, we fix any two states $\rho, \sigma \in \D{d}$ (we assume without loss of generality that $\rho$ and $\sigma$ have full rank---the general case follows by continuity). We then define a function $\v : [0,1] \to \R$ as follows:
\begin{equation}
  \v(p) := f_c\of[\big]{p \rho + (1-p) \sigma}
  \label{eq:v}
\end{equation}
(note that $\v(p)$ implicitly depends also on $\cs$). Our goal now is to determine the range of values of $\cs$ for which
\begin{equation}
  \v''(p) \leq 0
  \qquad \forall \, p \in [0,1] \text{\quad and \quad} \forall \, \rho, \sigma \in \D{d}.
  \label{eq:v''}
\end{equation}
This would imply that $\v(p)$ is concave and, in particular, that $\v(p) \geq p \v(1) + (1-p) \v(0)$, which by \cref{eq:v} is equivalent to concavity of $f_\cs$. The following lemma uses this approach to obtain the desired condition on
$(h,\cs)$.

\begin{lemma}\label{lem:concave-cond1}
Let $h: \R \to \R$ be any twice-differentiable, monotonously increasing function. Then the function $f_\cs(\rho) := h(\cs H(\rho))$ with $c \geq 0$ and $\rho \in \D{d}$ is concave on the set of quantum states $\D{d}$ if, for any probability distribution $q = (q_1, \dotsc, q_d)$, the following condition is satisfied:
\begin{equation}
  \cs \, \frac{h''(\cs H(q))}{h'(\cs H(q))} \le \frac{1}{\L(q)-H(q)^2}
  \label{eq:condition q}
\end{equation}
where $H(q) = - \sum_{i=1}^d q_i \log q_i$ is the Shannon entropy of $q$ and $\L(q):= \sum_{i=1}^d q_i (\log q_i)^2$.
\end{lemma}

To prove this lemma we employ the following definitions and results from~\cite{KA}.
For operators $A, \Delta \in \cL(\cH)$, where $A > 0$ and $\Delta$ is Hermitian, the Fr\'echet derivative of the operator logarithm is given by the linear, completely positive map $\Delta \mapsto \cT_A(\Delta)$~\cite{KA}, where
\begin{align}
  \cT_A(\Delta)
  &:= \frac{d}{dt} \bigg\vert_{t=0} \log (A+t\Delta), \nonumber \\
  & = \int_0^\infty ds (A+sI)^{-1} \Delta (A+sI)^{-1}. \label{Frechet integral}
\end{align}
Here the second line follows from the integral representation of the operator logarithm,
\begin{equation}
  \log A = \int_{0}^\infty ds \of*{\frac{1}{1+s}I - (A+sI)^{-1}}
  \quad \text{for any } A > 0,
\end{equation}
and the fact that
\begin{equation}
  \frac{d}{dt} \of{A + t\Delta}^{-1} = - \of{A + t\Delta}^{-1} \Delta \of{A + t\Delta}^{-1}.
\end{equation}
When $A$ and $\Delta$ commute, the integral in \cref{Frechet integral} can be worked out and we get $\cT_A(\Delta) = \Delta A^{-1}$.

It is easy to check that the map $\cT_A(\Delta)$ is self-adjoint, i.e., for any $B \in \cL(\cH)$,
\begin{equation}\label{Ts-a}
  \Tr \of[\big]{B \cT_A(\Delta)} = \Tr \of[\big]{\cT_A(B)\Delta},
\end{equation}
and that
\begin{equation}\label{T-1}
  \cT_A(A) = I.
\end{equation}
This linear map induces a metric on the space of Hermitian matrices given by
\begin{equation}\label{metric}
  M_A(\Delta) := \Tr \of[\big]{\Delta \cT_A(\Delta)}.
\end{equation}
This metric is known to be monotone~\cite{KA}; that is, for any completely positive trace-preserving linear map $\Lambda$,
\begin{equation}\label{mono}
  M_{\Lambda(A)} \of[\big]{\Lambda(\Delta)} \le M_A(\Delta).
\end{equation}

Now we are ready to prove \cref{lem:concave-cond1}. Our proof will proceed in two steps: first we will reduce the problem from general quantum states to commuting ones, and then restate the concavity condition for commuting states in terms of a similar condition for probability distributions.

\begin{proof}[Proof of \cref{lem:concave-cond1}]
Let $\Delta := \rho - \sigma$ and $\xi := p \rho + (1-p) \sigma = \sigma + p \Delta$. Note that $\xi' := \frac{d}{dp} \xi = \Delta$ and $\xi'' = 0$. Recall from \cref{eq:v''} that concavity of $f_\cs$ is equivalent to $\v''(p) \leq 0$ where
\begin{equation}
  \v(p):= f_\cs(\xi) = h(\cs H(\xi)).
  \label{eq:v(p)}
\end{equation}
To compute $\v''(p)$, we will need to find the first two derivatives of $H(\xi) = - \Tr (\xi \log \xi)$ with respect to $p$. Noting that
\begin{equation}
  \frac{d}{dp} \log \xi = \cT_\xi(\xi')
  \label{eq:dlog}
\end{equation}
and using \cref{eq:dlog}, we find that the first derivative of $H(\xi)$ is
\begin{align}
  \smash[b]{\frac{d}{dp}} H(\xi)
  &= - \Tr\of[\big]{\xi' \log \xi} - \Tr\of[\big]{\xi \cT_\xi(\xi')} \nonumber \\
  &= - \Tr\of[\big]{\xi' \log \xi} - \Tr\of[\big]{\cT_\xi(\xi) \xi'} \nonumber \\
  &= - \Tr\of[\big]{\xi' \log \xi} - \Tr \xi' \nonumber \\
  &= - \Tr\of[\big]{\xi' \log \xi}. \label{eq:H1}
\end{align}
In the first line we used the Fr\'echet derivative of the logarithm as given in \cref{eq:dlog}, while the second line follows from the self-adjointness \eqref{Ts-a} of the map $\cT_\xi$. The last two lines follow from \cref{T-1} and the fact that $\Tr \xi' = \Tr \Delta = 0$. The second derivative is
\begin{equation}
  \smash[b]{\frac{d^2}{dp^2}} H(\xi)
  = - \Tr\of[\big]{\xi'' \log \xi} - \Tr\of[\big]{\xi' \cT_\xi(\xi')}
  = - M_\xi(\xi'),
  \label{eq:H2}
\end{equation}
where the first term vanishes since $\xi'' = 0$ while the second term produces $M_\xi(\xi')$ by \cref{metric}.

We are now ready to calculate the second derivative of $\v(p) = h(\cs H(\xi))$ introduced in \cref{eq:v(p)}. By the chain rule,
\begin{align}
  \v' (p) &= \cs h'\of[\big]{\cs H(\xi)} \frac{dH(\xi)}{dp}, \\
  \v''(p) &= \cs^2 h''\of[\big]{\cs H(\xi)} \sof*{\frac{dH(\xi)}{dp}}^2
          + \cs   h' \of[\big]{\cs H(\xi)} \frac{d^2 H(\xi)}{dp^2}.
\end{align}
Therefore, $\v''(p) \leq 0$ is equivalent to
\begin{equation}
  \cs h''\of[\big]{\cs H(\xi)} \sof*{\Tr\of{\xi' \log \xi}}^2
 \leq h' \of[\big]{\cs H(\xi)} M_\xi(\xi'),
\end{equation}
where we divided by $c > 0$ (the case $c=0$ is trivial) and substituted the derivatives of $H(\xi)$ from \cref{eq:H1,eq:H2}. Since we imposed the condition that $h$ is monotonously increasing, we can divide by $h'$ and get the condition
\begin{equation}
  \cs \,
  \frac{h''\of[\big]{\cs H(\xi)}}
       {h' \of[\big]{\cs H(\xi)}}
  \leq
  \frac{M_\xi(\Delta)}
       {\sof{\Tr(\Delta\log\xi)}^2}.
  \label{eq:KA1}
\end{equation}

By fixing the state $\xi$ and minimizing the right-hand side over all $\Delta$, we get a stronger inequality, which in particular implies \cref{eq:KA1}. Consider the dephasing channel $\Lambda := \diag_\xi$ which, when acting on an operator $\Delta$, sets all its off-diagonal elements equal to $0$ in any basis in which $\xi$ is diagonal (in particular, in its eigenbasis). Thus, $\diag_\xi(\xi) = \xi$ and
\begin{equation}\label{mono-diag}
  M_{\xi}\of[\big]{\diag_\xi(\Delta)} \le M_\xi(\Delta),
\end{equation}
by the monotonicity property \eqref{mono} of the metric $M_\xi(\Delta)$ under CPTP maps. Hence, on replacing $\Delta$ by $\diag_\xi(\Delta)$ on the right-hand side of \cref{eq:KA1}, the denominator remains the same but the numerator does not increase. Since $[\diag_\xi(\Delta), \xi] = 0$, to obtain the minimum value of the right-hand side of \cref{eq:KA1}, it therefore suffices to restrict to those $\Delta$ which commute with $\xi$.

Recall that $\cT_\xi(\Delta) = \Delta \xi^{-1}$ for commuting $\xi$ and $\Delta$, so
\begin{equation}
  M_\xi(\Delta)
  = \Tr\of[\big]{\Delta \cT_\xi(\Delta)}
  = \sum_{i=1}^d \delta_i^2/\xi_i
\end{equation}
where $\xi_i$ and $\delta_i$ for $i \in \set{1, \dotsc, d}$ are the diagonal elements of $\xi$ and $\Delta$ in the eigenbasis of $\xi$ (in fact, $\xi_i$ are the eigenvalues of $\xi$). We can now phrase the problem of minimizing the right-hand side of \cref{eq:KA1} as follows:
\begin{equation}
  \text{minimize} \quad
  \frac{\sum_{i=1}^d \delta_i^2/\xi_i}{\of{\sum_{i=1}^d \delta_i \log \xi_i}^2} \quad
  \text{subject to} \quad
  \sum_{i=1}^d \delta_i = 0,
  \label{eq:frac min}
\end{equation}
where the condition $\sum_{i=1}^d \delta_i = 0$ arises from the fact that $\Tr \Delta = 0$.

Since the objective function in \cref{eq:frac min} is invariant under scaling of all $\delta_i$ by the same scale factor, we can convert the minimization problem to the following one:
\begin{equation}
  \text{minimize}   \quad \sum_{i=1}^d \delta_i^2/\xi_i \quad
  \text{subject to} \quad \sum_{i=1}^d \delta_i = 0
  \;\text{ and } \sum_{i=1}^d \delta_i \log \xi_i = 1.
  \label{eq:min}
\end{equation}
Using the method of Lagrange multipliers, we form the Lagrangian
\begin{equation}
  \cL
  := \sum_{i=1}^d \delta_i^2/\xi_i
   - 2 \lambda \sum_{i=1}^d \delta_i
   - 2 \mu \of*{ \sum_{i=1}^d \delta_i \log \xi_i - 1 }.
\end{equation}
To find its stationary points, we require that $\partial\cL/\partial\delta_i = 0$ for all $i$. This implies
\begin{equation}
  \delta_i = \xi_i (\lambda + \mu \log \xi_i).
  \label{eq:deltas}
\end{equation}
To find the Lagrange multipliers $\lambda$ and $\mu$, we substitute the $\delta_i$ back into the constraints of the optimization problem \eqref{eq:min}. We get the following equations:
\begin{align}
  \lambda   - \mu H &= 0, &
 -\lambda H + \mu \L &= 1,
\end{align}
where $H := - \sum_{i=1}^d \xi_i \log \xi_i$ and $\L := \sum_{i=1}^d \xi_i (\log \xi_i)^2$. Their solution is
\begin{align}
  \mu &= \frac{1}{\L-H^2}, &
  \lambda &= \frac{H}{\L-H^2}.
  \label{eq:mu lambda}
\end{align}

Inserting \cref{eq:deltas,eq:mu lambda} back in the objective function of \cref{eq:min} yields
\begin{align}
    \sum_{i=1}^d \delta_i^2/\xi_i
 &= \sum_{i=1}^d \xi_i (\lambda + \mu \log \xi_i)^2 \nonumber \\
 &= \lambda^2 - 2 \lambda \mu H + \mu^2 \L \nonumber \\
 &= \frac{1}{\L-H^2}.
\end{align}
Thus, \cref{eq:KA1} is satisfied whenever
\begin{equation}
  \cs \, \frac{h''(\cs H)}{h'(\cs H)} \leq \frac{1}{\L-H^2}.
  \label{eq:condition}
\end{equation}
Note that $H = H(q)$ and $\L = \L(q)$ where $q := (q_1, \dots, q_d)$ with $q_i := \xi_i$ is a probability distribution. Thus, condition \eqref{eq:condition q} implies \cref{eq:condition} and hence the concavity of $f_\cs$.
\end{proof}

The quantity $L - H^2$ arising on the right-hand side of \cref{eq:condition} is known as the \emph{variance of the surprisal} $(- \log q_i)$~\cite{RW15,PPV10}:
\begin{equation}
  V(q)
  := \L(q) - H(q)^2
   = \sum_{i=1}^d q_i (-\log q_i)^2
   - \of[\bigg]{\sum_{i=1}^d q_i (-\log q_i)}^2.
\end{equation}
To find the optimal value of $\cs$ for which \cref{eq:condition} holds, we need to minimize its right-hand side over all attainable values of the quantity $\L - H^2$ for a fixed value of $H$. In other words, we require the maximum attainable value of $\L(q) - H(q)^2$ over all probability distributions $q$ over $d$ elements with a fixed value of the entropy $H(q) = H_0$ (in contrast, ref. \cite{RW15} evaluated the maximum value of $V(q)$ \emph{without} the constraint of $H(q)$ being fixed). We define
\begin{equation}
  \L_{\max}(H_0) :=
  \max \set[\Big]{\L(q) : H(q) = H_0, \, \sum_{i=1}^d q_i = 1, \, \text{and $q_i \geq 0$ for all $i$}}.
  \label{eq:Vmax def}
\end{equation}
To obtain this value and the corresponding optimal distribution $q$, we employ the following lemma.

\begin{lemma}\label{lem:abb}
The maximum of $\L(q) := \sum_{i=1}^d q_i (\log q_i)^2$ over all probability distributions $q = (q_1, \dotsc, q_d)$ with fixed Shannon entropy $H(q) = H_0 \in [0, \log d]$ is achieved by a distribution of the form
\begin{equation}
  q = \of{x,\dotsc,x,y}
  \text{ for some } 0 \leq x < y \text{ such that } (d-1) x + y = 1.
\end{equation}
If we let $r := d-1$, then the value of $\L(q)$ achieved by this distribution is
\begin{equation}
  \L_{\max}(H_0) = rx (1-rx) \of[\big]{\log x - \log(1-rx)}^2 + H_0^2.
\end{equation}
\end{lemma}

\begin{proof}
For given $H_0 \in [0,\log d]$, we need to solve the following constrained optimization problem:
\begin{equation}
  \text{maximize}   \quad \sum_{i=1}^d q_i (\log q_i)^2 \quad
  \text{subject to} \quad \sum_{i=1}^d q_i = 1
  \;\text{ and } - \sum_{i=1}^d q_i \log q_i = H_0.
  \label{eq:max}
\end{equation}
Since the domain of the logarithm is $\R^+$, we do not have to explicitly impose the condition that $q_i \geq 0$ for all $1 \leq i \leq d$.

The maximum of a continuously differentiable function $f$ over a domain $D$ either occurs at a stationary point of $f$, or on the boundary of $D$. In the present case $D$ is the probability simplex, hence its boundary consists of probability vectors where some of the $q_i$ are zero. Due to the fact that both $-q_i\log q_i$ and $q_i (\log q_i)^2$ are zero for $q_i=0$, such points can be conveniently modeled by treating them as probability vectors in a lower-dimensional probability space. We can therefore safely assume that the sought-after maximum occurs at the relative interior of a $K$-dimensional probability simplex (with $K \leq d$), and at the very end of the calculation perform a further maximization over $K$. In particular, it will turn out that the global maximum occurs for $K=d$.

The aforementioned maximum can be found as a stationary point of the Lagrangian
\begin{equation}
  \cL
  := \sum_{i=1}^K q_i (\log q_i)^2
   + \lambda \of[\bigg]{\sum_{i=1}^K q_i - 1}
   - \mu \of[\bigg]{\sum_{i=1}^K q_i \log q_i + H_0}.
\end{equation}
Requiring that all derivatives $\partial \cL / \partial q_i$ be zero yields the equations
\begin{equation}
  (\log q_i)^2 + (2-\mu) \log q_i + \lambda - \mu = 0.
\end{equation}
As this is a fixed \emph{quadratic} function of $\log q_i$, and therefore may have at most two solutions, we infer that the stationary points of $\cL$ are those distributions $q$ whose elements are either all equal (and hence equal to $1/K$) or equal to two possible values. That is, up to permutations, the distribution $q$ can be uniquely represented as
\begin{equation}
  q_{k,x} := (\underbrace{x, \dotsc, x}_{k}, \underbrace{y, \dotsc, y}_{K-k})
  \label{eq:q}
\end{equation}
for some integer $k \in \set{1, \dotsc, K}$ and some probabilities $0 \leq x < y$ such that
\begin{equation}
  k x + (K-k) y = 1.
  \label{eq:Normalization}
\end{equation}
From this we get in addition that $x \leq 1/K < y$. For $k = K$, there is only one distribution of this form, namely, the uniform distribution $q_{d,x} = (1/K, \dotsc, 1/K)$. This distribution has $H(q_{K,x}) = \log K$ and $\L(q_{K,x}) = (\log K)^2$, which are independent of $x$, so there is nothing to optimize in this case.

\begin{figure}[t]
\centering
\includegraphics[width=12cm]{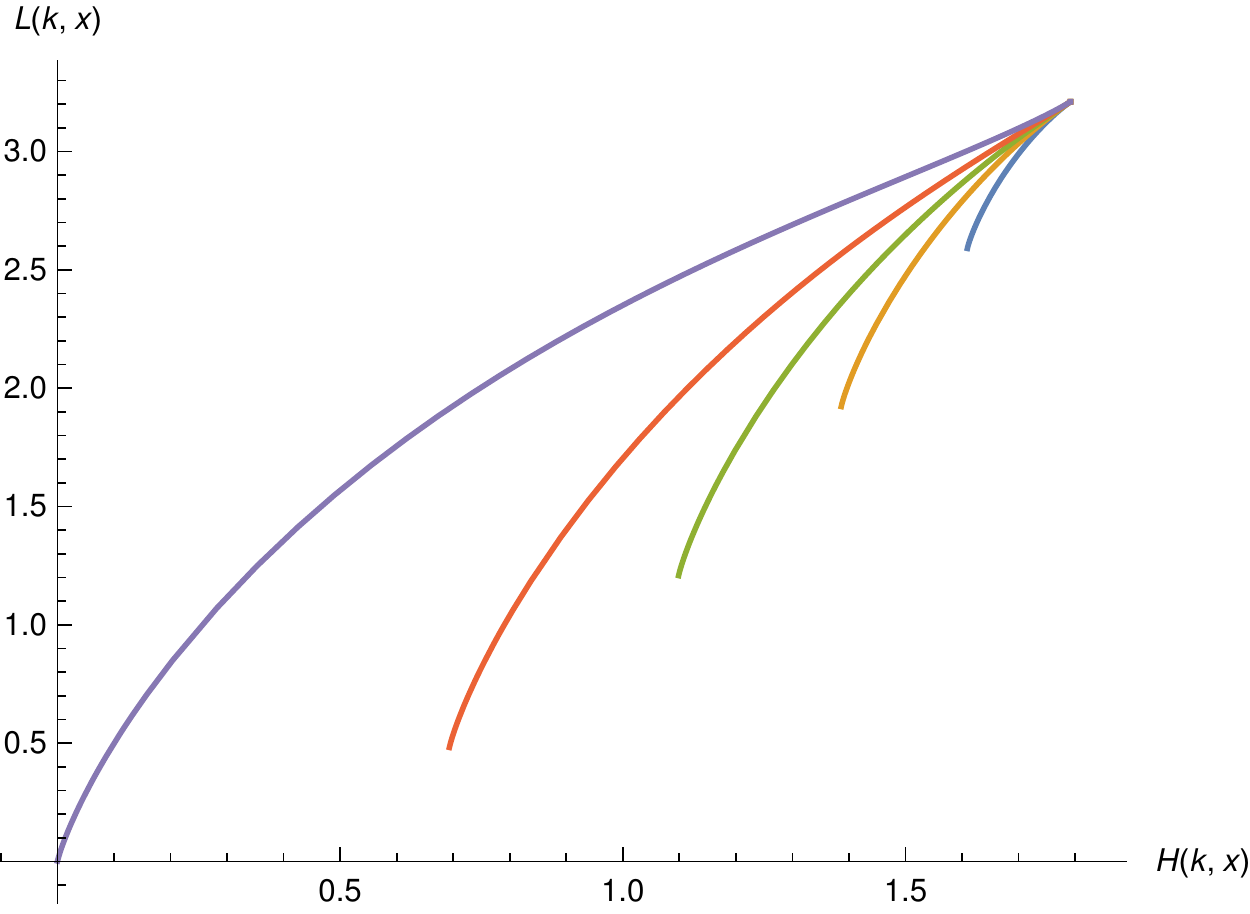}
\caption{The locus of the points $(H,\L)$ as $x$ varies over $[0,1/K]$, for the case $K=6$ and for each value of $k \in \set{1,\dotsc,5}$, with $k$ increasing towards the left. These loci are curves with lower end point $H=H(k,0)=\log (K-k)$, $\L=\L(k,0)=(\log(K-k))^2$ and upper end point $H=H(k,1/K)=\log K$, $\L=\L(k,1/K)=(\log K)^2$. The value $k=0$ yields a single point, just as the $k=K$ case does, coinciding with the upper end point of all other $(H,\L)$ curves.\label{fig:KA1}}
\end{figure}

From now on we assume that $k \neq K$ and thus $H_0 < \log K$. Then $y := (1 - k x) / (K - k)$ from the normalization constraint \eqref{eq:Normalization}, so we can compute
\begin{align}
  H(k,x) &:=  H(q_{k,x}) = -k x  \log x    - (K-k) y  \log y,    \label{eq:Hkx} \\
 \L(k,x) &:= \L(q_{k,x}) =  k x (\log x)^2 + (K-k) y (\log y)^2. \label{eq:Vkx}
\end{align}
To obtain the global maximum of $\L(k,x)$, a further optimization over $k \in \set{1, \dotsc, K-1}$ and $x \in [0, 1/K)$ is required. The numerical calculations presented in the diagram in \cref{fig:KA1} suggest that $k=K-1$ yields the maximal value of $\L$. To prove that this is actually true we will temporarily remove the restriction that $k$ be an integer and consider the entire range $k\in(0,K)$. Our analysis will show that keeping $H(k,x)$ fixed, $\L(k,x)$ increases with $k$.

\newcommand{\kxk}{\of[\big]{k,x(k)}}

To keep $H$ fixed as $k$ changes, $x$ will have to change as well. For given $H_0 \in [0, \log K)$, let $x(k)$ be the function of $k$ implicitly given by $H(k,x(k)) = H_0$. We would like to know how $\L(k,x(k))$ changes as a function of $k$.
Taking the total derivative with respect to $k$ gives
\begin{align}
      \frac{d}{dk}H\kxk
   &= \frac{\partial}{\partial k}H\kxk
    + \frac{\partial}{\partial x}H\kxk \, x'(k) = 0, \\
      \frac{d}{dk}\L\kxk
   &= \frac{\partial}{\partial k}\L\kxk
    + \frac{\partial}{\partial x}\L\kxk \, x'(k).
\end{align}
Solving the first equation for $x'(k)$ and substituting the solution in the second equation gives
\begin{align}
  \frac{d}{dk}\L\kxk
  &= \frac{\partial}{\partial k}\L\kxk
   - \frac{\partial}{\partial x}\L\kxk \left.
     \frac{\partial}{\partial k}H\kxk \middle/
     \frac{\partial}{\partial x}H\kxk \right. \nonumber \\
  &= \frac{1}{K-k}
     \sof*{ \of[\Big]{1+(K-2k)x} \of[\Big]{\log \frac{1-kx}{K-k} - \log x} - 2(1-Kx)}, \label{632}
\end{align}
where the second line follows by substituting the partial derivatives of $H(k,x)$ and $\L(k,x)$ defined in \cref{eq:Hkx,eq:Vkx}.

Note that $z \geq \tanh z = (e^{2z}-1)/(e^{2z}+1)$ for any $z \geq 0$. By choosing $z := (\log w) / 2$ we get $\log w \geq 2(w-1)/(w+1)$ for $w \geq 1$. Next, since $x \leq 1/K$, we can take $w := (1-kx)/(Kx-kx) \geq 1$ which gives
\begin{equation}
  \log \frac{1 - k x}{K - k} - \log x
  \geq 2 \left.
    \of*{\frac{1-kx}{(K-k)x}-1} \middle/
    \of*{\frac{1-kx}{(K-k)x}+1} \right.
  = 2 \, \frac{1-Kx}{1+(K-2k)x}.
\end{equation}
Inserting this in \cref{632} and noting that $1+(K-2k)x\ge0$ for $x\le1/K$, we conclude that $d\L(k,x(k)) / dk \geq 0$, so $\L(k,x(k))$ is increasing as a function of $k$ just as we intended to show.

Reverting back to integer values of $k$, we find that, for a fixed value of $H(k,x)$, the value of $\L(k,x)$ is maximized when $k$ is the largest integer in the open interval $(0,K)$, namely $K-1$. Then $\L_{\max}(H_0)$, the maximum value of $\L(k,x)$ subject to $H(k,x) = H_0$, see \cref{eq:Vmax def}, is given by $\L(K-1,x)$ where $x$ is such that $H(K-1,x) = H_0$.

Finally, we have to perform a further maximization over $K$, for $K\le d$. In a similar way as before $x$ becomes a function of $K$. We now show that the maximum of $\L(K-1,x(K))$ under the constraints $H(K-1,x(K))=H_0$ and $K\le d$ occurs for $K=d$. Solving the equation $0 = \frac{d}{dK}H(K-1,x(K))$ for $x'(K)$ and substituting the solution back into $\frac{d}{dK}\L(K-1,x(K))$ shows after a fair bit of algebra that
\begin{equation}
  \frac{d}{dK}\L(K-1,x(K)) = x(K) \sof[\Big]{2+\log\of[\big]{1-K+1/x(K)}},
\end{equation}
which is clearly non-negative for $0 \le x(K) \le 1/K$, hence $\L(K-1,x(K))$ increases with $K$. We conclude that the overall maximum occurs for $K=d$, as we set out to prove.

The last statement of the lemma is now easily shown. From \cref{eq:Vkx} we infer that
\begin{equation}
  \L_{\max}(H_0) - H_0^2
  = \L(d-1,x) - H(d-1,x)^2
  = rx (1-rx) \of[\big]{\log x - \log(1-rx)}^2,
\end{equation}
where $r := d-1$ and $x \in [0, 1/d]$ satisfies $H(d-1,x) = H_0$.
\end{proof}

For $r = d-1$ and any $x \in [0,1/d]$, if $q_{r,x}$ is the probability distribution defined in \cref{eq:q}, we denote its Shannon entropy and the information variance by
\begin{align}
  s_r(x) &:= H(q_{r,x}) = - rx \log x - (1-rx) \log(1-rx) \label{s-r-x}, \\
  w_r(x) &:= V(q_{r,x}) = \L(q_{r,x}) - H(q_{r,x})^2 = rx (1-rx) \of[\big]{\log x - \log(1-rx)}^2. \label{w-r-x}
\end{align}
In terms of these quantities, the condition in \cref{lem:concave-cond1}, under which a given function of the von~Neumann entropy is concave on the set of qudit states, is expressed by the following theorem.

\begin{theorem}\label{thm:concave-cond2}
Let $h: \R \to \R$ be a twice-differentiable, monotonously increasing function. Then the function $f_{\cs}(\rho) := h(\cs H(\rho))$ with $\cs \geq 0$ and $\rho \in \D{d}$ is concave on $\D{d}$ if
\begin{equation}
  \cs \, \frac{h''(\cs s_r(x))}{h'(\cs s_r(x))} \leq \frac{1}{w_r(x)},
  \label{condition x}
\end{equation}
for all $0 \leq x \leq 1/d$, where $r:=d-1$ and functions $s_r(x)$ and $w_r(x)$ are defined in \cref{s-r-x,w-r-x}.
\end{theorem}

\subsection{Concavity of Entropy Power}

In this section we use \cref{thm:concave-cond2} to establish the first item of \cref{thm:Concavity}, namely, that the entropy power
$\EP_c(\rho) = e^{\cs H(\rho)}$ of a state $\rho \in \D{d}$ is concave for
$0 \leq \cs \leq 1/(\log d)^2.$

\begin{proof}[Proof of \cref{thm:Concavity} (concavity of $\EP_\cs$)]
In this case we have $h(x)=\exp(x)$, so the condition \eqref{condition x} just translates to
\begin{equation}
  \cs\le 1/w_{d-1}(x), \quad 0\le x\le 1/d.
\end{equation}
Therefore,
\begin{equation}\label{eq:cmax}
  \cs \le \of[\Big]{ \max_{0\le x\le 1/d} w_{d-1}(x) }^{-1} =: \cs_{\max}.
\end{equation}
From the expression of $w_r$ follows a simple lower bound on the largest allowed value of $\cs$.
Putting $y=rx$, with $0\le y\le (d-1)/d<1$,
\begin{align}
  w_r(x)
 &= y(1-y) \of[\big]{-\log r+\log y-\log(1-y)}^2 \nonumber \\
 &= A (\log r)^2 + B \log r + C,
\end{align}
where the coefficients $A,B,C$ of this quadratic polynomial in $\log r$ are bounded above as follows:
$A := y(1-y)\le 1/4$,
$B := -2y(1-y) \of[\big]{\log y-\log(1-y)} \le 1/2$, and
$C := y(1-y) \of[\big]{\log y-\log(1-y)}^2 \le 1/2$.
Hence,
\begin{equation}
  w_r(x)\le (1+\log r)/2+(\log r)^2/4 \quad \text{with } r=d-1,
\end{equation}
and we obtain
\begin{equation}\label{cmax bound}
  \cs_{\max} \ge 1 / \sof[\big]{\of[\big]{1+\log (d-1)}/2 + \of[\big]{\log(d-1)}^2/4}.
\end{equation}
This bound becomes asymptotically exact in the limit of large $d$. Note that the right-hand side of \cref{cmax bound} is larger than $ 1/(\log d)^2$ for $d \geq 3$. For $d = 2$, the right-hand side of \cref{cmax bound} is equal to 2 which is not larger than $1/(\log 2)^2 \approx 2.0814$. However, for this case one can numerically evaluate the expression \eqref{eq:cmax} for $\cs_{\max}$ to obtain the value $2.2767$, which is indeed greater than $1/(\log 2)^2$.
\end{proof}

From this we can also infer that for any probability distribution $p$ over $d$ elements, the function $E(p) := e^{c H(p)}$ is concave for $0 \leq c \leq 1/(\log d)^2$.

\begin{remark}
The fact that for $h(x) := \exp(x)$ the inequality~\eqref{eq:condition q} of \cref{lem:concave-cond1} holds for any value of $c$ in the range $0 \leq c \leq 1/(\log d)^2$ can also be proved using Theorem~8 and Lemma~15 of~\cite{RW15}.
\end{remark}

\subsection{Concavity of Entropy Photon Number}

In this section we use \cref{thm:concave-cond2} to establish the second item of \cref{thm:Concavity}, namely, that the entropy photon number $N_\cs(\rho)$ of a qudit, defined by \cref{def:EPn}, is concave for $0 \leq \cs \leq 1/(d-1)$.

\begin{proof}[Proof of \cref{thm:Concavity} (concavity of $N_\cs$)]
In this case the calculations are more complicated because $h$ is not given directly but as the inverse of a function: $h=g^{-1}$, where
\begin{equation}
  g(x)=-x\log(x)+(1+x)\log(1+x).
\end{equation}
The derivatives of $h$ are given by
\begin{align}
  h'(x) &= \frac{1}{g'(h(x))} = \frac{1}{\log(1+1/h(x))}, \\
  h''(x) &= \frac{1}{\of[\big]{h(x)+h^2(x)} \sof[\big]{\log(1+1/h(x))}^3}.
\end{align}
Defining the function
\begin{equation}
  k(x) = x(1+x) \of[\big]{\log(x)-\log(1+x)}^2,
\end{equation}
we have $h'(x)/h''(x) = k(h(x))$. The function $k$ is monotonously increasing, concave, and ranges from 0 to 1. The condition on $\cs$ becomes
\begin{equation}
  g^{-1}(\cs s_r(x)) \ge k^{-1}(\cs w_r(x)).
\end{equation}
If we define the variables $y$ and $z$ according to
\begin{equation}
  g(y)=\cs s_r(x), \qquad k(z)=\cs w_r(x),
\end{equation}
the condition is $y\ge z$. If we now exploit the monotonicity of $k$, this condition is equivalent to $k(y)\ge k(z)=\cs w_r(x)$. We therefore require that
\begin{equation}
  g(y)=\cs s_r(x) \,\,\mbox{ implies } \,\, k(y)\ge \cs w_r(x).\label{eq:KAcondphot}
\end{equation}
We will show that this holds for $\cs\le 1/a=1/(d-1)$. In \cref{fig:KA2} we depict the graph of $k(y)$ versus $g(y)$. The graph seems to indicate that the resulting curve is concave and monotonously increasing; that this is actually true follows from the easily checked fact that the function $k'/g'=(1+2x)\of[\big]{\log(1+x)-\log x}-2$, representing the slope of the curve, is positive and decreasing. The condition \eqref{eq:KAcondphot} amounts to the statement that any point $(\cs s_r(x),\cs w_r(x))$ lies in the area below this curve. Hence if the condition is satisfied for a certain value of $\cs$, then it is also satisfied for any smaller positive value of $\cs$. Therefore, we only need to prove \cref{eq:KAcondphot} for $\cs=1/(d-1)$.

\begin{figure}[t]
\centering
\includegraphics[width=8cm]{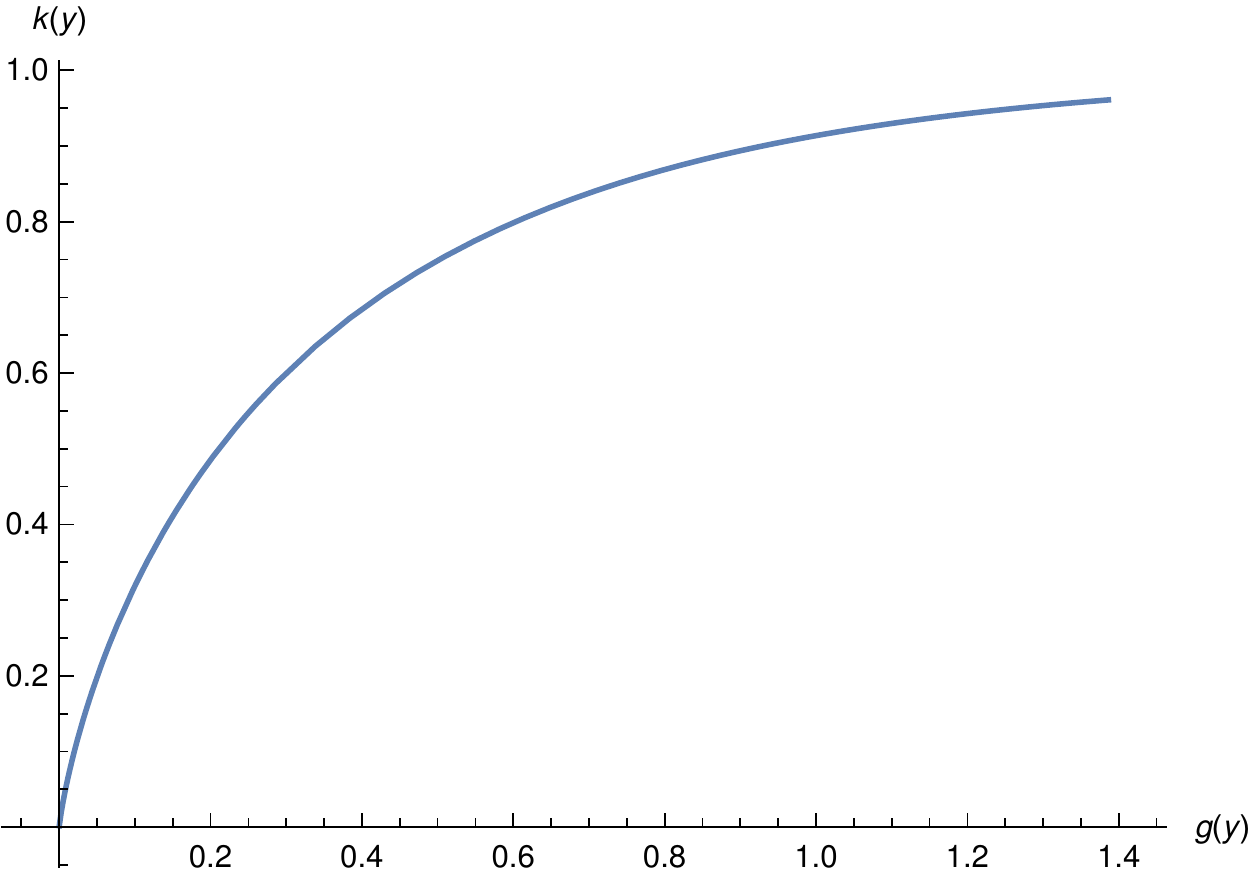}
\caption{Parametric plot of $k(y)$ versus $g(y)$.\label{fig:KA2}}
\end{figure}

\newcommand{\dom}{\mathcal{S}}

The formal similarities between $g$ and $s_r$ and between $k$ and $w_r$ let us define two interpolating functions $g_1(x,b)$ and $k_1(x,b)$ as a function of the original $x$ and an interpolation parameter $b$:
\begin{align}
  g_1(x,b) &= -x\log x + (1+bx)\frac{\log(1+bx)}{b}, \\
  k_1(x,b) &= x(1+bx)\of[\big]{\log x-\log(1+bx)}^2.
\end{align}
Let $\dom$ indicate the domain of $g_1$ and $k_1$, which is $b\in [1-d,1]$ and $x\in[0,1/d]$, as before. To ensure continuity of $g_1$ at $b=0$, we define $g_1(x,0)$ to be its limit value $\of{-x\log x+x}$. Hence, we have the correspondences
\begin{align}
  s_r(x)/(d-1) &= g_1(x,1-d),&\quad& g(x) = g_1(x,1),\\
  w_r(x)/(d-1) &= k_1(x,1-d),&\quad& k(x) = k_1(x,1).
\end{align}
The condition \eqref{eq:KAcondphot} is therefore satisfied if a continuous path $x(b)$ exists (from $x(1-d)=x$ to $x(1)=y$) such that $g_1(x(b),b)$ remains constant and $k_1(x(b),b)$ increases with $b$. As in the proof of \cref{lem:abb} this requires the positivity of
\begin{align}
  \frac{d}{db} k_1(x(b),b)
  &=\frac{\partial}{\partial b}k_1(x(b),b)
  - \frac{\partial}{\partial x}k_1(x(b),b) \left. \frac{\partial}{\partial b}g_1(x(b),b) \middle/ \frac{\partial}{\partial x}g_1(x(b),b) \right. \nonumber \\
  &= \frac{1}{b^2}\Bigl[
     bx \of[\big]{2+\log x+bx(\log x)^2} \nonumber \\
     &\quad +(1+bx)^2 \of[\big]{\log(1+bx)}^2 \nonumber \\
     &\quad -\of[\Big]{2+bx+(1+2bx+2b^2x^2)\log x} \log(1+bx)
     \Bigr].
\end{align}
Let us introduce the variable $u=1+bx$. In $\dom$ we have $(1-b)x\le 1$ so that $x\le u$; furthermore, $b\le 1$ and $x\le 1/d$, so that $u\le 1+1/d$. The second factor can now be written more succinctly as
\begin{align}
  \MoveEqLeft 
  (u-1)\left(2+\log x +(u-1)(\log x)^2\right)+u^2(\log u)^2
  -(1+u+(1-2u+2u^2)\log x)\log u\nonumber\\
  ={}& (u-1)^2\log x(\log x-\log u) + u^2(\log u)^2 \nonumber \\
   &+ 2(u-1)-(u+1)\log u-(1-u+u^2\log u)\log x.
\end{align}
The first two terms are clearly non-negative. The factor $1-u+u^2\log u$ is non-negative too, as can be seen from the inequality $1-\exp(-v)\le v\le v\exp(v)$ applied to $v=\log u$. Furthermore, $\log x\le \log(1/d)\le \log(1/2)\le -1/2$, so that the last term is bounded below by $(1-u+u^2\log u)/2$. It is therefore left to show that $2(u-1)-(u+1)\log u+(1-u+u^2\log u)/2$ is non-negative.

For $0<u\le 1$ we can exploit the inequality $\log u\le 2(u-1)/(u+1)$, so that we obtain $2(u-1)-(u+1)\log u\ge0$. The remaining term is non-negative too, as we have just showed.

For $1\le u\le 1+1/d$ we exploit instead the inequality $\log u\le u-1$. Then based on the fact that in this range $(u-1)^2-3<0$
\begin{align}
  \MoveEqLeft 2(u-1)-(u+1)\log u+(1-u+u^2\log u)/2 \nonumber \\
  &=   \frac{1}{2} \of[\Big]{3(u-1)+\of[\big]{(u-1)^2-3} \log u} \nonumber \\
  &\ge \frac{1}{2} \of[\Big]{3(u-1)+\of[\big]{(u-1)^2-3}(u-1)} \nonumber \\
  &=   (u-1)^3/2\ge0.
\end{align}

This shows that $k_1(x(b),b)$ indeed increases with $b$, whence condition \eqref{eq:KAcondphot} holds for $\cs=1/(d-1)$ and, by a previous argument, for $\cs\le1/(d-1)$. In other words, we have shown that the function $g^{-1}(\cs H(\rho))$ is concave for $0<\cs\le 1/(d-1)$. As this includes the value $\cs=1/d$,
the photon number is concave.
\end{proof}

\section{Bounds on minimum output entropy and Holevo capacity} \label{sec:Application}

As an application of our results we now consider the class of quantum channels $\E_{\a,\sigma} : \D{d} \to \D{d}$ obtained from the partial swap channel $\E_\a$ from \cref{eq:E} by fixing the second input state $\sigma$ (see \cref{fig:Channel}). Such channels are parameterized by a variable $\a \in [0,1]$ \emph{and} a quantum state $\sigma \in \D{d}$, and act as follows:
\begin{equation}
  \E_{\a,\sigma}(\rho) := \rho \boxplus_\a \sigma.
  \label{eq:Eas}
\end{equation}

\begin{figure}[t]
\centering
\input{fig-channel.tex}
\caption{A schematic representation of the channel $\E_{\a,\sigma}$ defined in \cref{eq:Eas}.\label{fig:Channel}}
\end{figure}
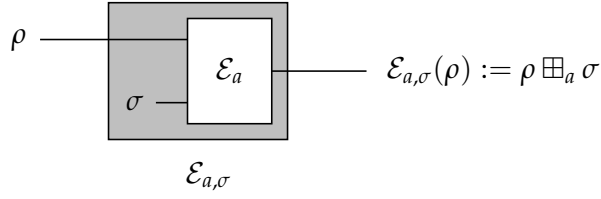

For example, for the choice $\sigma = I/d$ (the completely mixed state) the channel $\E_{\a,\sigma}$ is just the quantum depolarizing channel with parameter $\a$. If $\sigma = \delta \proj{0} + (1-\delta) \proj{1} \in \D{2}$ for some $\delta \in [0,1]$, then $\E_{\a,\sigma}$ is a qubit channel whose output density matrix is
\begin{equation}
  \mx{
    \a r_{00} + (1-\a) \delta & r_{01} (\a - i \sqrt{\a(1-\a)} (1-2\delta)) \\
    r_{01} (\a + i \sqrt{\a(1-\a)} (1-2\delta)) & \a r_{11} + (1-\a) (1-\delta)
  }
\end{equation}
for any input qubit state $\rho := \sum_{i,j=0}^1 r_{ij} \ket{i}\bra{j}$.

An important characteristic quantity for any quantum channel $\E$ is its \emph{minimum output entropy}, which is defined as
\begin{equation}\label{moe}
  H_{\min}\of{\E} := \min_\rho H\of{\E(\rho)}.
\end{equation}
Lower bounds on this quantity for the class of channels $\E_{\a,\sigma}$ can be obtained by using our EPIs and EPnI. In fact, the inequalities of \cref{thm:EPIs} give various lower bounds on the output entropy of the channel $\E_{\a,\sigma}$ (i.e. the entropy of any output state) in terms of the entropy $H(\rho)$ of an input state $\rho$:
\begin{align}
  H\of{\E_{\a,\sigma}(\rho)}
    &\geq \a H(\rho) + (1-\a) H(\sigma), \label{use} \\
  H\of{\E_{\a,\sigma}(\rho)}
    &\geq \frac{1}{c} \log \sof[\Big]{\a\exp(cH(\rho)) + (1-\a) \exp(cH(\sigma))},
    && \text{with } c = 1/(\log d)^2, \label{use2} \\
  H\of{\E_{\a,\sigma}(\rho)}
    &\geq \frac{1}{c} g \sof[\Big]{\a g^{-1}(cH(\rho)) + (1-\a) g^{-1}(cH(\sigma))},
    && \text{with } c = 1/(d-1). \label{use3}
\end{align}
Since the above bounds are of the form  $H\of{\E_{\a,\sigma}(\rho)} \geq G\of{H(\rho)}$, for some function $G$, we have
\begin{align}
  H_{\min}\of{\E_{\a,\sigma}} &\geq \min_\rho G\of{H(\rho)}\nonumber\\
  &= \min_{0\leq H_0 \leq \log d} G(H_0).
  \label{eq:G}
\end{align}

In \cref{fig:KA3} we have plotted the bounds $G(H_0)$ for two illustrative cases, the three curves corresponding to the three choices of the function $G$ as given by the right-hand sides of \cref{use,use2,use3}. For the qubit ($d=2$) case we actually have a tight lower bound
\begin{equation}
  H\of{\E_{\a,\sigma}(\rho)}
  \geq \f \sof[\Big]{ \a \f^{-1}\of[\big]{H(\rho)} + (1-\a) \f^{-1}\of[\big]{H(\sigma)} }
  \label{eq:Optimal d=2}
\end{equation}
where $\f(r)$ is the entropy of a qubit state whose Bloch vector has length $r$, see \cref{eq:H and f} in \cref{apx:EPI for qubits}. This bound follows from \cref{eq:Tight} in \cref{apx:EPI for qubits} and is also shown in \cref{fig:KA3}.

\begin{figure}[ht]
\centering
\includegraphics[width=8cm]{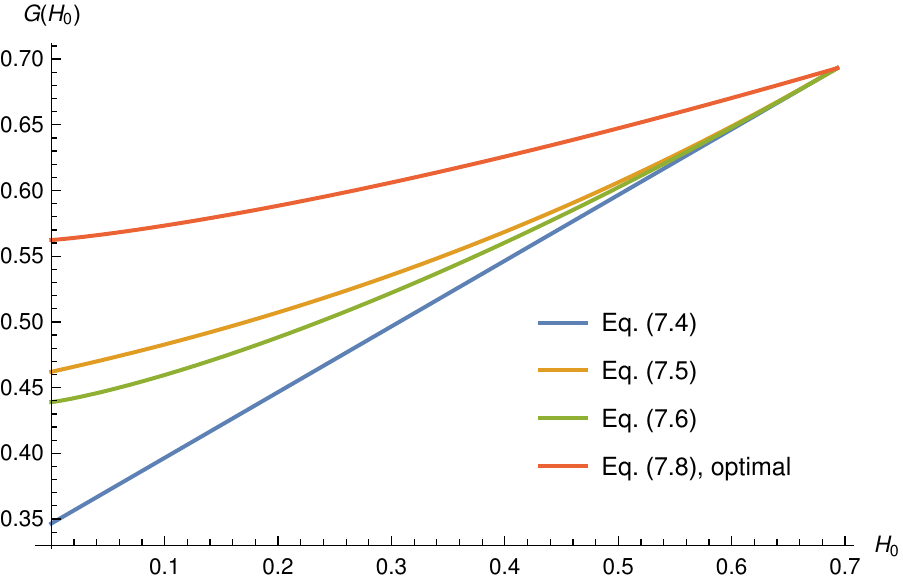}\quad
\includegraphics[width=8cm]{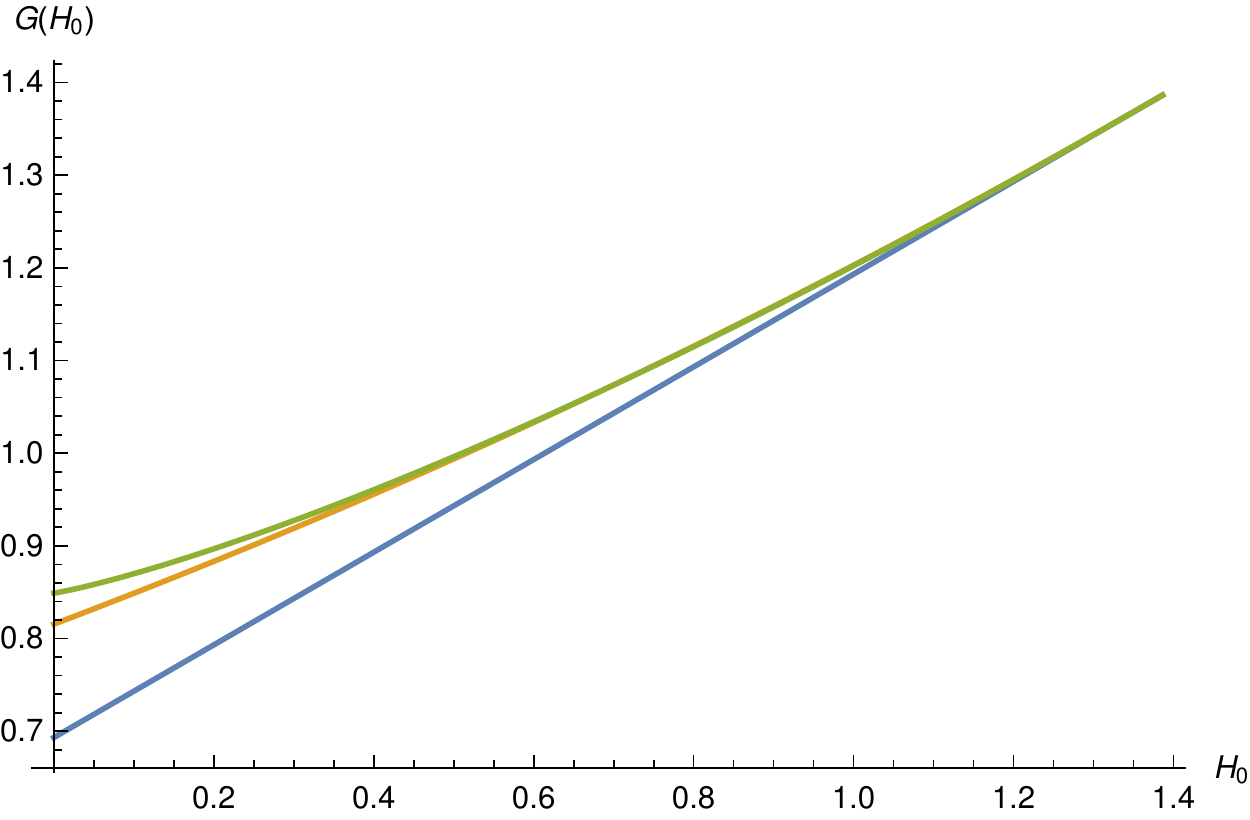}
\caption{Plots of bounds $G$ from \cref{eq:G} for the channel $\E_{1/2,\sigma}(\rho)$, where $\sigma$ is the maximally mixed state $\sigma=I/d$ in dimensions $d=2$ (left panel) and $d=4$ (right panel). The blue curves represent the bound \eqref{use} obtained from \cref{EPI H}, the orange curves represent the bound \eqref{use2} obtained from the entropy power inequality \eqref{EPI exp}, and the green curves represent the bound \eqref{use3} obtained from the entropy photon number inequality \eqref{EPnI-d}. For $d=2$, the optimal bound \eqref{eq:Optimal d=2} is given by the pink curve in the left panel. While neither of the bounds in~\cref{use,use2,use3} is optimal for this channel, the numerics suggest that the entropy photon number inequality is the best out of the three when $d \geq 4$. For $d = 2$, however, the entropy power inequality \eqref{use2} yields a better bound.\label{fig:KA3}}
\end{figure}

These bounds imply lower bounds on the minimum output entropy $H_{\min}\of{\E_{a,\sigma}}$, which in turn allow us to obtain upper bounds on the \emph{product-state classical capacity} of $\E_{a,\sigma}$. The latter is the capacity evaluated in the limit of asymptotically many independent uses of the channel, under the constraint that the inputs to multiple uses of the channel are necessarily product states. The Holevo-Schumacher-Westmoreland (HSW)~\cite{Holevo, SW97} theorem establishes that the product-state capacity of a memoryless quantum channel $\E$ is given by its Holevo capacity $\chi(\E)$:
\begin{equation}
  \chi(\E):= \max_{\{p_i, \rho_i\}}
  \set[\bigg]{H \of[\Big]{\sum_i p_i \E(\rho_i)} - \sum_i p_i H\of[\big]{\E(\rho_i) }},
\end{equation}
where the maximum is taken over all ensembles $\{p_i, \rho_i\}$ of possible input states $\rho_i$ occurring with probabilities $p_i$. Using the above expression, and the fact that $H(\omega) \le \log d$ for any $\omega \in \D{d}$, we obtain the following simple bound:
\begin{equation}
  \chi(\E) \le \log d - \min_{\rho} H\of[\big]{\E(\rho)},
\end{equation}
where the minimum is taken over all possible inputs to the channel. Applying this bound to the channel $\E_{\a,\sigma}$ for any $\a \in [0,1]$ and $\sigma \in \D{d}$ and using \cref{use} we infer that
\begin{align}
  \chi(\E_{\a,\sigma})
  &\le \log d - \a \min_\rho H(\rho) - (1-\a) H(\sigma) \nonumber \\
  &= \log d - (1-\a) H(\sigma).
\end{align}
For the case of the qubit channel introduced above, we thus obtain the bound
\begin{equation}
\chi(\E_{\a,\sigma}) \le \log 2 - (1-a)h(\delta),
\end{equation}
where $h(\delta):= - \delta \log \delta  - (1-\delta) \log (1-\delta)$ is the binary entropy. Even sharper bounds are possible by exploiting \cref{use2,use3}.

\section{Summary and open questions}

In this paper we establish a class of entropy power inequalities (EPIs) for $d$-level quantum systems or qudits. The underlying addition rule for which these inequalities hold, is given by a quantum channel acting on the product state $\rho \otimes \sigma$ of two qudits and yielding the state of a single qudit as output. We refer to this channel as a \emph{partial swap channel} since its output interpolates between the states $\rho$ and $\sigma$ as the parameter $a$ on which it depends is changed from $1$ to $0$. We establish EPIs not only for the von~Neumann entropy and the entropy power, but also for a large class of functions, which include the R\'enyi entropies and the subentropy. Moreover, for the subclass of partial swap channels for which one of the qudit states in the input is fixed, our EPI for the von~Neumann entropy yields an upper bound on the Holevo capacity.

We would like to emphasize that the method that we employ to prove our EPIs is novel, in the sense that it does not mimic the proofs of the EPIs in the continuous-variable classical and quantum settings. Instead it relies solely on spectral majorization and concavity of certain functions.

\subsection{Open questions}

Our results lead to many interesting open questions; here we briefly mention some of them. For example, can a conditional version of the EPI (see~\cite{Koenig}) be proved for qudits? Can an optimal bound similar to \cref{eq:Optimal d=2} be found also for $d > 2$? Is it possible to generalize our quantum addition rule~\eqref{eq:Box} for combining more than two states? Such a generalization has recently been obtained for three states~\cite{Ozols15}, though the problem for four or more states is not yet fully resolved. More importantly, proving analogues of our EPI for three or more states (similar to the multi-input EPI of \cite{DMLG14}) remains an interesting open question. Finally, is the partial swap channel that we define the unique channel resulting in an interpolation between the input states and yielding a non-trivial EPI (i.e., one that is not simply a statement of concavity)? According to~\cite{Ozols15}, it is unique (up to the sign of $i$) in a certain class of channels.

In \cref{sec:Application}, we mentioned a simple application of our EPI to quantum Shannon theory. Considering the significance of the classical EPI in information theory and statistics, we expect that our EPIs will also find further applications.

Finally, it would be worth exploring whether our proof of the qudit analogue of the entropy photon number inequality can be generalized to establish the EPnI for the bosonic case (which is known to be an important open problem).

\section*{Acknowledgements}

We would like to thank Jianxin Chen, Robert K\"onig and Will Matthews for useful discussions. We are grateful to David Reeb for correcting a typo in the previous version of our paper and for pointing us to the paper \cite{RW15} where an optimization problem similar to the one we consider in \cref{sec:Concavity} was studied. We would also like to thank an anonymous referee for helpful suggestions that improved our paper and for pointing out the optimality of the bound in \cref{eq:Optimal d=2}. KA acknowledges support by an Odysseus Grant of the Flemish FWO. MO acknowledges financial support from European Union under project QALGO (Grant Agreement No.~600700) and by a Leverhulme Trust Early Career Fellowhip (ECF-2015-256).


\bibliographystyle{alphaurl}
\bibliography{References}

\appendix

\section{Entropy power inequality for qubits} \label{apx:EPI for qubits}

For the case of qubits ($d=2$), there is a simple proof of \cref{EPI H} which exploits the Bloch-vector representation of a qubit state.

\subsection{Qubit states and the Bloch sphere} \label{apx:Bloch sphere}

It is known that the state $\rho$ of a qubit can be expressed in terms of its Bloch vector $\vec{r}$ as follows:
\begin{equation}
  \rho = \frac{1}{2} (I + \vec{r} \cdot \vec{\sigma})
       = \frac{1}{2} (I + x \sigma_x + y \sigma_y + z \sigma_z),
  \label{eq:Qubit rho}
\end{equation}
where $\vec{r} := (x,y,z) \in \R^3$ such that $\abs{\vec{r}} := \sqrt{x^2 + y^2 + z^2} \leq 1$. Here $\vec{r} \cdot \vec{\sigma}$ denotes a formal inner product between $\vec{r}$ and $\vec{\sigma} := (\sigma_x, \sigma_y, \sigma_z)$, with $\sigma_x$, $\sigma_y$ and $\sigma_z$ being the Pauli matrices. Moreover, the eigenvalues of the state $\rho$ can easily be seen to be given by $\frac{1}{2} (1 \pm \abs{\vec{r}})$. Hence, its von~Neumann entropy is simply
\begin{equation}
  H(\rho) = h\of[\big]{ \tfrac{1}{2}(1 + \abs{\vec{r}}) },
  \label{eq:Qubit H}
\end{equation}
where $h(p) := - p \log p - (1-p) \log (1-p)$ is the \emph{binary entropy} of $p \in [0,1]$ in nats. For $x \in [-1,1]$, let us define the function
\begin{equation}
  \f(x) := h\of[\big]{ \tfrac{1}{2}(1 + x) }.
  \label{eq:f}
\end{equation}
One can easily see that $\f$ is symmetric around the vertical axis and verify that
\begin{equation}
  \f''(x) = -\frac{1}{1-x^2} \leq 0,
  \label{eq:f''}
\end{equation}
so $\f$ is concave (see \cref{fig:Plot of f}). In terms of this function, \cref{eq:Qubit H} is given by
\begin{equation}
  H(\rho) = \f(\abs{\vec{r}}).
  \label{eq:H and f}
\end{equation}

\begin{figure}
\centering
\input{fig-g.tex}
\caption{A plot of the function $\f$ defined in \cref{eq:f}.\label{fig:Plot of f}}
\end{figure}
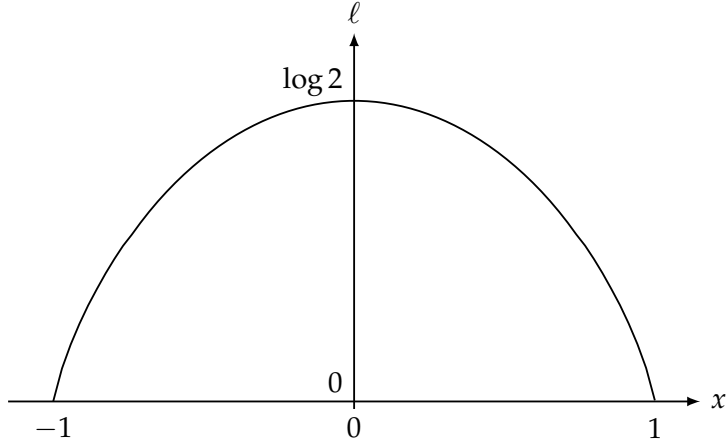

\subsection{Proof of the qubit EPI} \label{apx:Proof of qubit EPI}

For a pair of qubit states $\rho_1$ and $\rho_2$, the first EPI of \cref{thm:EPIs} is given by
\begin{equation}
  H\of[\big]{\rho_1 \boxplus_\a \rho_2}
  \geq \a H(\rho_1) + (1-\a) H(\rho_2), \qquad \forall \, \a \in [0,1].
\label{eq:qubits}
\end{equation}
Below is a simple proof of the above inequality for the special case of qubits.

\begin{proof}
Using \cref{eq:H and f}, the inequality \eqref{eq:qubits} can be expressed in terms of the function $\f$ as follows:
\begin{equation}
  \f(r) \geq \a \f(r_1) + (1-\a) \f(r_2),
  \label{eq:EPI for f}
\end{equation}
where $r := \abs{\vec{r}}$, $r_1 := \abs{\vec{r}_1}$, $r_2 := \abs{\vec{r}_2}$, and $\vec{r}$, $\vec{r}_1$, $\vec{r}_2$ denote the Bloch vectors of the states $\rho_1 \boxplus_\a \rho_2$, $\rho_1$ and $\rho_2$, respectively. Recall from \cref{eq:r} that $\vec{r}$ can be expressed in terms of $\vec{r}_1$ and $\vec{r}_2$ as follows:
\begin{equation}
  \vec{r}
  = \a \vec{r}_1
  + (1-\a) \vec{r}_2
  + \sqrt{\a(1-\a)} (\vec{r}_1 \times \vec{r}_2).
\end{equation}
Since $\vec{r}_1 $ and $\vec{r}_2$ are both perpendicular to $\vec{r}_1 \times \vec{r}_2$, we get
\begin{equation}
  \abs{\vec{r}}^2
  = \vec{r} \cdot \vec{r}
  = \a^2 \abs{\vec{r}_1}^2
  + (1-\a)^2 \abs{\vec{r}_2}^2
  + 2 \a(1-\a) \vec{r}_1 \cdot \vec{r}_2
  + \a(1-\a) \abs{\vec{r}_1 \times \vec{r}_2}^2.
  \label{eq:rr}
\end{equation}
If we denote by $\gamma \in [0,\pi]$ the angle between vectors $\vec{r}_1$ and $\vec{r}_2$, then $\vec{r}_1 \cdot \vec{r}_2 = \abs{\vec{r}_1} \abs{\vec{r}_2} \cos \gamma$ and $\abs{\vec{r}_1 \times \vec{r}_2} = \abs{\vec{r}_1} \abs{\vec{r}_2} \sin \gamma$, so \cref{eq:rr} becomes
\begin{equation}
  r^2
  = \a^2 r_1^2
  + (1-\a)^2 r_2^2
  + \a(1-\a) \of{2 r_1 r_2 \cos \gamma + r_1^2 r_2^2 \sin^2 \gamma}.
  \label{eq:r2}
\end{equation}

Note that the right-hand side of the inequality \eqref{eq:EPI for f} does not depend on the angle $\gamma$ between the vectors $\vec{r}_1$ and $\vec{r}_2$, so it suffices to prove \cref{eq:EPI for f} only for those values of $\gamma$ that \emph{minimize} the left-hand side. Since $f(r)$ is a decreasing function of $r$ for $r \geq 0$ (see \cref{fig:Plot of f}), we have to consider only those values of $\gamma$ that \emph{maximize} $r$. From \cref{eq:r2} we have that
\begin{equation}
  r = \sqrt{
        \a^2 r_1^2
        + (1-\a)^2 r_2^2
        + \a(1-\a) r_1 r_2 \of{2 \cos \gamma + r_1 r_2 \sin^2 \gamma}
      }
\end{equation}
where $\a, r_1, r_2 \in [0,1]$. To maximize this over $\gamma$, we only need to maximize the last term. Note that
\begin{equation}
       2 \cos \gamma + r_1 r_2 \sin^2 \gamma
  \leq 2 \cos \gamma + \sin^2 \gamma
  \leq 2,
\end{equation}
where the last inequality is tight if and only if $\gamma = 0$. This gives a simple upper bound on $r$:
\begin{equation}
  r \leq \sqrt{ \a^2 r_1^2 + (1-\a)^2 r_2^2 + 2 \a(1-\a) r_1 r_2}
  = \a r_1 + (1-\a) r_2.
\end{equation}

Since $\f(r)$ is monotonically decreasing for $r \geq 0$, we get
\begin{equation}
  \f(r) \geq \f \of[\big]{ \a r_1 + (1-\a) r_2 }.
  \label{eq:Tight}
\end{equation}
Note that this lower bound is independent of the parameter $\gamma$ and is tight (it becomes equality when $\gamma = 0$). Recall from \cref{eq:f''} that $f$ is concave (see also \cref{fig:Plot of f}), so
\begin{equation}
  \f \of[\big]{ \a r_1 + (1-\a) r_2 }
  \geq \a \f(r_1) + (1-\a) \f(r_2).
\end{equation}
By combining the last two inequalities, we get the desired result.
\end{proof}

\end{document}

%% file: fig-beamsplitter.tex

\begin{tikzpicture}[
  semithick,
  ray/.style = {-latex}]

  \begin{scope}[xshift = -4cm]
    \def\d{0.3}
    \node at (-2, 1+\d) {$\hat{a}$};
    \node at (-2,-1+\d) {$\hat{b}$};
    \node at ( 2, 1+\d) {$\hat{d}$};
    \node at ( 2,-1+\d) {$\hat{c}$};
    \node at ( 0, 1+\d) {$B_\a$};
    \draw[ray] (-3, 1) -- (-1, 1) -- (1,-1) -- (3,-1);
    \draw[ray] (-3,-1) -- (-1,-1) -- (1, 1) -- (3, 1);
    \draw (0,1) -- (1,0) -- (0,-1) -- (-1,0) -- cycle;
    \draw (-1,0) -- (1,0);
  \end{scope}

  \begin{scope}[xshift = 4cm]
    \def\d{0.3}
    \node at (-2, 1+\d) {$\rho_1$};
    \node at (-2,-1+\d) {$\rho_2$};
    \draw (-3, 1) -- (3, 1);
    \draw (-3,-1) -- (3,-1);
    \node[draw, fill = white, minimum height = 2.7cm, minimum width = 1.3cm] at (0,0) {$U_\a$};
  \end{scope}

%
%
%
%

\end{tikzpicture}

%% file: fig-channel.tex

\begin{tikzpicture}[semithick]

  \def\zx{1.4}
  \def\zy{1.4}

  \def\h{0.3*\zy}

  \node[draw, fill = lightgray, minimum height = \zy*1.3cm, minimum width = \zx*1.7cm] at (-0.3*\zx,0) {};

  \draw (0,0) -- (1.3*\zx,0);
  \node at (2.5*\zx,0) {$\E_{\a,\sigma}(\rho) := \rho \boxplus_\a \sigma$};

  \draw (-1.8*\zx, \h) -- (0, \h);
  \draw (-0.7*\zx,-\h) -- (0,-\h);
  \node at (-2.0*\zx, \h) {$\rho$};
  \node at (-0.9*\zx,-\h) {$\sigma$};

  \node[draw, fill = white, minimum height = \zy*1cm, minimum width = \zx*0.8cm] at (0,0) {$\E_\a$};

  \node at (-0.2*\zx,-1*\zy) {$\E_{a,\sigma}$};

\end{tikzpicture}

%% file: fig-g.tex

\def\zx{4.0}
\def\zy{4.0}

\begin{tikzpicture}
  [domain = -0.9999:0.9999, variable = \x, samples = 70,
   line width = 0.7pt,
   arc/.style = {-latex}]

  \def\d{0.1}
  \def\D{0.6}

  \draw[arc] (-\zx-\D,0) -- (\zx+\D,0) node[right] {$x$};
  \draw[arc] (0,-\d) -- (0,\zy+1.5*\D) node[above] {$\f$};

  \draw plot (\zx*\x, { -\zy * ifthenelse(\x, ((1+\x)/2)*log2((1+\x)/2) + ((1-\x)/2)*log2((1-\x)/2), 1) });

  \def\l{0.35}

  \draw (-\zx,-\l) node {$-1$};
  \draw (0   ,-\l) node {$0$};
  \draw ( \zx,-\l) node {$1$};
  \draw (-\l,  0)+(+0.1,0.25) node {$0$};
  \draw (-\l,\zy)+(-0.2,0.25) node {$\log 2$};

\end{tikzpicture}